\documentclass[a4paper,UKenglish,cleveref, autoref]{lipics-v2021}

\pdfoutput=1 
\hideLIPIcs  

\title{On the Weihrauch degree of the additive Ramsey theorem}

\author{Arno Pauly}{School of Mathematics and Computer Science, Swansea University, UK}{arno.m.pauly@gmail.com}{https://orcid.org/0000-0002-0173-3295}{}
\author{Cécilia Pradic}{School of Mathematics and Computer Science, Swansea University, UK}{c.pradic@swansea.ac.uk}{https://orcid.org/0000-0002-1600-8846}{}
\author{Giovanni Sold\`a}{School of Mathematics and Computer Science, Swansea University, UK \footnote{Sold\`a has since moved to Ghent University}}{giovanni.a.solda@gmail.com}{https://orcid.org/0000-0003-4903-3623}{This author was supported by an LMS Early Career Fellowship.}

\authorrunning{A.~Pauly, C.~Pradic \& G.~Sold\`a} 

\Copyright{Arno Pauly, Cécilia Pradic and Giovanni Sold\`a} 



\category{} 

\relatedversion{This article extends the conference paper \cite{cieversion} by the second and third author.} 

\funding{All authors acknowledge support by Swansea University.}

\nolinenumbers 

\usepackage[T1]{fontenc}
\def\doi#1{\href{https://doi.org/\detokenize{#1}}{\url{https://doi.org/\detokenize{#1}}}}
\usepackage{graphicx}
\usepackage{hyperref}
\usepackage{todonotes}
\usepackage[utf8]{inputenc}
\usepackage{stmaryrd}
\usepackage{amsmath}
\usepackage{amssymb}
\usepackage{amsfonts}
\usepackage{thmtools}
\usepackage{array}
\usepackage{xspace}
\usepackage[all]{xy}
\usepackage{cmll}
\usepackage{nth}

\newcommand\pow{\mathcal{P}}
\newcommand\powfin{\mathcal{P}_{\mathrm{fin}}}

\newcommand{\ident}[2]{\newcommand{#1}{\ensuremath{\texorpdfstring{\mathrm{#2}}{#2}}\xspace}}

\ident{\dom}{dom}
\ident{\rg}{rg}
\ident{\sgn}{sgn}

\newcommand{\cA}{\mathcal{A}}
\newcommand{\cS}{\mathcal{S}}
\newcommand{\cB}{\mathcal{B}}
\newcommand{\cN}{\mathcal{N}}
\newcommand{\cC}{\mathcal{C}}
\newcommand{\cG}{\mathcal{G}}
\newcommand{\cL}{\mathcal{L}}
\newcommand{\cE}{\mathcal{E}}
\newcommand{\cP}{\mathcal{P}}
\newcommand{\cR}{\mathcal{R}}
\newcommand{\cD}{\mathcal{D}}

\newcommand{\bbN}{\mathbb{N}}

\newcommand{\cO}{\mathcal O}
\newcommand{\cF}{\mathcal F}

\newcommand{\bbQ}{\mathbb{Q}}

\newcommand{\eqdef}{\mathrel{\mathop{:}=}}

\newcommand\restr[2]{{
  \left.\kern-\nulldelimiterspace 
  #1 
  \vphantom{\big|} 
  \right|_{#2} 
  }}

\newcommand{\card}[1]{{| #1 |}}

\newcommand{\GrJ}{\mathrel{\mathcal{J}}}
\newcommand{\GrH}{\mathrel{\mathcal{H}}}
\newcommand{\GrL}{\mathrel{\mathcal{L}}}
\newcommand{\GrR}{\mathrel{\mathcal{R}}}

\newcommand{\ind}{\textrm{-}\mathsf{IND}}

\newcommand{\rca}{\mathsf{RCA}_0}

\newcommand{\rt}{\mathsf{RT}}

\newcommand{\ooint}[1]{{] {#1} [}}

\newcommand{\Nn}{\bN}

\renewcommand\vec\overline

\newcommand\subto\rightarrowtail

\newcommand\sztindtxt{$\Sigma^0_2$-induction}
\newcommand\sztind{\Sigma^0_2\ind}

\newcommand\wft{\mathsf}
\newcommand\weiortom{\wft{ORT}}

\newcommand\weishuffle{\wft{Shuffle}}

\newcommand\colA{c}

\ifdefined\cA
\else
\newcommand\cA{\mathcal{A}}
\fi
\ifdefined\cB
\else
\newcommand\cB{\mathcal{B}}
\fi
\ifdefined\cC
\else
\newcommand\cC{\mathcal{C}}
\fi
\ifdefined\cD
\else
\newcommand\cD{\mathcal{D}}
\fi
\ifdefined\cE
\else
\newcommand\cE{\mathcal{E}}
\fi
\ifdefined\cF
\else
\newcommand\cF{\mathcal{F}}
\fi
\ifdefined\cG
\else
\newcommand\cG{\mathcal{G}}
\fi
\ifdefined\cH
\else
\newcommand\cH{\mathcal{H}}
\fi
\ifdefined\cI
\else
\newcommand\cI{\mathcal{I}}
\fi
\ifdefined\cJ
\else
\newcommand\cJ{\mathcal{J}}
\fi
\ifdefined\cK
\else
\newcommand\cK{\mathcal{K}}
\fi
\ifdefined\cL
\else
\newcommand\cL{\mathcal{L}}
\fi
\ifdefined\cM
\else
\newcommand\cM{\mathcal{M}}
\fi
\ifdefined\cN
\else
\newcommand\cN{\mathcal{N}}
\fi
\ifdefined\cO
\else
\newcommand\cO{\mathcal{O}}
\fi
\ifdefined\cP
\else
\newcommand\cP{\mathcal{P}}
\fi
\ifdefined\cQ
\else
\newcommand\cQ{\mathcal{Q}}
\fi
\ifdefined\cR
\else
\newcommand\cR{\mathcal{R}}
\fi
\ifdefined\cS
\else
\newcommand\cS{\mathcal{S}}
\fi
\ifdefined\cT
\else
\newcommand\cT{\mathcal{T}}
\fi
\ifdefined\cU
\else
\newcommand\cU{\mathcal{U}}
\fi
\ifdefined\cV
\else
\newcommand\cV{\mathcal{V}}
\fi
\ifdefined\cW
\else
\newcommand\cW{\mathcal{W}}
\fi
\ifdefined\cX
\else
\newcommand\cX{\mathcal{X}}
\fi
\ifdefined\cY
\else
\newcommand\cY{\mathcal{Y}}
\fi
\ifdefined\cZ
\else
\newcommand\cZ{\mathcal{Z}}
\fi
\ifdefined\bA
\else
\newcommand\bA{\mathbb{A}}
\fi
\ifdefined\bB
\else
\newcommand\bB{\mathbb{B}}
\fi
\ifdefined\bC
\else
\newcommand\bC{\mathbb{C}}
\fi
\ifdefined\bD
\else
\newcommand\bD{\mathbb{D}}
\fi
\ifdefined\bE
\else
\newcommand\bE{\mathbb{E}}
\fi
\ifdefined\bF
\else
\newcommand\bF{\mathbb{F}}
\fi
\ifdefined\bG
\else
\newcommand\bG{\mathbb{G}}
\fi
\ifdefined\bH
\else
\newcommand\bH{\mathbb{H}}
\fi
\ifdefined\bI
\else
\newcommand\bI{\mathbb{I}}
\fi
\ifdefined\bJ
\else
\newcommand\bJ{\mathbb{J}}
\fi
\ifdefined\bK
\else
\newcommand\bK{\mathbb{K}}
\fi
\ifdefined\bL
\else
\newcommand\bL{\mathbb{L}}
\fi
\ifdefined\bM
\else
\newcommand\bM{\mathbb{M}}
\fi
\ifdefined\bN
\else
\newcommand\bN{\mathbb{N}}
\fi
\ifdefined\bO
\else
\newcommand\bO{\mathbb{O}}
\fi
\ifdefined\bP
\else
\newcommand\bP{\mathbb{P}}
\fi
\ifdefined\bQ
\else
\newcommand\bQ{\mathbb{Q}}
\fi
\ifdefined\bR
\else
\newcommand\bR{\mathbb{R}}
\fi
\ifdefined\bS
\else
\newcommand\bS{\mathbb{S}}
\fi
\ifdefined\bT
\else
\newcommand\bT{\mathbb{T}}
\fi
\ifdefined\bU
\else
\newcommand\bU{\mathbb{U}}
\fi
\ifdefined\bV
\else
\newcommand\bV{\mathbb{V}}
\fi
\ifdefined\bW
\else
\newcommand\bW{\mathbb{W}}
\fi
\ifdefined\bX
\else
\newcommand\bX{\mathbb{X}}
\fi
\ifdefined\bY
\else
\newcommand\bY{\mathbb{Y}}
\fi
\ifdefined\bZ
\else
\newcommand\bZ{\mathbb{Z}}
\fi

\DeclareMathOperator{\ran}{\mathrm{ran}}

\newcommand{\rra}{\rightrightarrows}
\newcommand{\sierp}{\mathbb{S}}
\newcommand{\baire}{\bbN^\bbN}
\newcommand{\cantor}{\mathbf{2}^\bbN}

\newcommand{\weipr}[1]{\mathsf{#1}}

\newcommand{\id}{\mathsf{id}}
\newcommand{\lpo}{\mathsf{LPO}}
\newcommand{\tcn}{\mathsf{TC}_{\mathbb{N}}}
\newcommand{\ect}{\mathsf{ECT}}
\newcommand{\cn}{\mathsf{C}_\bbN}
\newcommand{\Cc}[1]{\mathsf{C}_{#1}}

\newcommand{\isfin}{\mathsf{IsFinite_\sierp}}
\newcommand{\sisfin}{\mathsf{IsFinite}}
\newcommand{\siscofin}{\mathsf{IsCofinite}}
\newcommand{\CCantor}{\Cc{\cantor}}

\newcommand{\ort}{\mathsf{ORT}}

\newcommand{\art}{\mathsf{ART}}

\newcommand{\etap}{(\eta)^1_{<\infty}}
\newcommand{\cetap}{\mathsf{c}(\eta)^1_{<\infty}}
\newcommand{\ietap}{\mathsf{i}(\eta)^1_{<\infty}}

\newcommand{\crt}{\mathsf{cRT}}

\newcommand{\cshuffle}{\mathsf{cShuffle}}
\newcommand{\ishuffle}{\mathsf{iShuffle}}
\newcommand{\sshuffle}{\mathsf{Shuffle}}

\newcommand{\cart}{\mathsf{cART}}
\newcommand{\iart}{\mathsf{iART}}
\newcommand{\sart}{\mathsf{ART}}

\newcommand{\cort}{\mathsf{cORT}}
\newcommand{\iort}{\mathsf{iORT}}

\newcommand{\ltW}{<_\mathrm{W}}
\newcommand{\leqW}{\leq_\mathrm{W}}
\newcommand{\leW}{\le_\mathrm{W}}

\newcommand{\equivW}{\equiv_\mathrm{W}}
\newcommand{\nleqW}{\nleq_\mathrm{W}}

\newcommand{\leqsW}{\leq_\mathrm{sW}}
\newcommand{\lesW}{\le_\mathrm{sW}}

\newcommand{\equivsW}{\equiv_\mathrm{sW}}

\newtheorem{question}[theorem]{Question}
\newcommand{\swanseastatement}{}

\EventEditors{John Q. Open and Joan R. Access}
\EventNoEds{2}
\EventLongTitle{42nd Conference on Very Important Topics (CVIT 2016)}
\EventShortTitle{CVIT 2016}
\EventAcronym{CVIT}
\EventYear{2016}
\EventDate{December 24--27, 2016}
\EventLocation{Little Whinging, United Kingdom}
\EventLogo{}
\SeriesVolume{42}
\ArticleNo{23}

\begin{document}

\maketitle

\begin{abstract}
  We characterize the strength, in terms of
  Weihrauch degrees, of certain problems related
  to Ramsey-like theorems concerning colourings of the rationals and of the natural numbers.
  The theorems we are chiefly interested in assert the existence
  of almost-homogeneous sets for colourings of pairs of rationals respectively natural numbers
  satisfying properties determined by some additional algebraic
  structure on the set of colours.

  In the context of reverse mathematics, most of the principles we study
  are equivalent to \sztindtxt\ over $\rca$. The associated problems in
  the Weihrauch lattice are related to $\tcn^*$, $(\lpo')^*$ or their product,
  depending on their precise formalizations.

\end{abstract}


\section{Introduction}

The infinite Ramsey theorem says that for any colouring $\colA$ of $n$-uples of a given arity of an infinite set $X$, there exists a infinite subset $H \subseteq X$ such that the
set of $n$-tuples $[H]^n$ of elements of $H$ is homogeneous. 
This statement is non-constructive: even if the colouring $\colA$ is given by a computable function, it is not the case that we can find a computable homogeneous subset of $X$.
Various attempts have been made to quantify how non-computable this problem and some of its natural restrictions are. This is in turn linked
to the axiomatic strength of the corresponding theorems, as investigated in \emph{reverse mathematics}~\cite{simpson} where Ramsey's theorem is a
privileged object of study~\cite{slicing-hirschfeldt}.


This paper is devoted to a variant of Ramsey's theorem with the following restrictions:
we colour pairs of rational numbers and we require some additional structure on the colouring,
namely that it is \emph{additive}. 
A similar statement first appeared in~\cite[Theorem 1.3]{shelah1975} to give a self-contained proof of decidability of
the monadic second-order logic of $(\bbQ,<)$. We will also analyse a simpler statement
we call the \emph{shuffle principle}, a related tool appearing in more modern decidability proofs~\cite[Lemma 16]{CCP11}.
The shuffle principle states that every $\bbQ$-indexed word (with letters in a finite alphabet)
contains a convex subword in which every letter appears densely or not at all.
Much like the additive restriction of the Ramsey theorem for pairs over $\bbN$, studied from the point of view of
reverse mathematics in~\cite{KMPS19},  we obtain a neat correspondence with \sztindtxt\ ($\sztind$).

\begin{theorem}\label{thm:q-eqs}
In the weak second-order arithmetic $\rca$, $\sztind$ is equivalent to both the shuffle principle and the additive Ramsey theorem for $\bbQ$.
\end{theorem}

We take this analysis one step further in the framework of Weihrauch reducibility
that allows to measure the uniform strength of general multi-valued functions (also called \emph{problems}) over
Baire space.
Let $\sshuffle$ and $\sart_\bbQ$ be the most obvious problems corresponding to
the shuffle principle and additive Ramsey theorem over $\bbQ$ respectively (see Definitions \ref{def:wei-shuffle} and \ref{def:sartbbq} for specifics).
We relate them, as well as various weakenings $\cshuffle$, $\cart_\bbQ$, $\ishuffle$ and $\iart_\bbQ$
that only output sets of colours or intervals, to the standard (incomparable) problems $\tcn$ and $\lpo'$. We also consider the \emph{ordered Ramsey principle}, $\ort_\bbQ$, where the colours $k$ come equipped with a partial order $\preceq$, and the colouring $\alpha : [\bbQ]^2 \to k$ satisfying that $\alpha(r_1,r_2) \preceq \alpha(q_1,q_2)$ if $q_1 \leq r_1 < r_2 \leq q_2$. A weakening of $\sshuffle$ is the principle $\etap$ introduced in \cite{FrittaionP17} where we ask merely for an interval where some colour is dense; respectively for a colour which is dense somewhere.

\begin{theorem}\label{thm:wei-eqs}
We have the following equivalences
\begin{itemize}
\item $\sshuffle \equivW \sart_\bbQ \equivW \tcn^* \times (\lpo')^*$
\item $\cshuffle \equivW \cart_\bbQ \equivW (\lpo')^*$
\item $\ishuffle \equivW \iart_\bbQ \equivW \etap \equivW \ietap \equivW \tcn^*$
\item $\ort_\bbQ \equivW \lpo^*$
\item $\cetap \equivW \crt^1_+$
\end{itemize}
\end{theorem}

Finally, we turn to carrying out the analysis of those Ramseyan theorems over
$\bbN$ in the framework of Weihrauch reducibility. The additive Ramsey theorem over $\bbN$ is also an important tool in the study of monadic second order logic over countable scattered orders. As for the case
of $\bbQ$, we relate problems $\art_\bbN$ and $\ort_\bbN$ as well
as some natural weakenings $\cart_\bbN$, $\cort_\bbN$, $\iart_\bbN$ and
$\iort_\bbN$, to $\tcn$ and $\lpo'$ (the $\mathsf{i}$ variants of those principle
return, rather than an interval, some upper bound $n$ on the first two points of
some infinite homogeneous set).

\begin{theorem} \label{thm:wei-eqs-bbN}
We have the following equivalences
\begin{itemize}
  \item $\ort_\bbN \equivW \art_\bbN \equivW \tcn^* \times (\lpo')^*$
  \item $\cort_\bbN \equivW \cart_\bbN \equivW (\lpo')^*$
  \item $\iort_\bbN \equivW \iart_\bbN \equivW \tcn^*$
\end{itemize}
\end{theorem}

A diagram summarizing the relationship between the various problems we are
studying is given in Figure~\ref{fig:summary}.

\begin{figure}[ht]
  \begin{center}
    \begin{tikzpicture}[scale=0.6]
      
      \node (lpos) at (0,1.5) {$\lpo^*\equivW \ort_{\bbQ}$};
      \node (crtp) at (0,3) {$\mathsf{cRT}^1_+\equivW \mathsf{c}(\eta)^1_{<\infty}$};
      \node (lpops) at (-5,4.5) {$(\lpo')^*\equivW \mathsf{cX}$};
      \node (cn) at (8,4.5) {$\cn$};
      \node (tcns) at (5,6) {$\tcn^*\equivW \ect \equivW (\eta)^1_{<\infty}\equivW \mathsf{c}(\eta)^1_{<\infty} \equivW \mathsf{iX}$};
      \node (prod) at (0,7.5) {$(\tcn\times\lpo')^*\equivW \mathsf{X}$};
      
      \draw [->] (crtp) to (lpos);
      \draw [->] (lpops) to (crtp);
      \draw [->] (cn) to (lpos);
      \draw [->] (tcns) to (cn);
      \draw [->] (prod) to (lpops);
      \draw [->] (prod) to (tcns);
      \draw [->] (tcns) to (crtp);
    \end{tikzpicture}
    \caption{Reductions and non-reductions between the problems studied in \Cref{thm:wei-eqs-bbN} and \Cref{thm:wei-eqs}. An arrow from $A$ to $B$ means that problem $B$ Weihrauch reduces to
    problem $A$, and lack of an arrow means lack of such a reduction (except when implied by transitivity). The problem $\mathsf{X}$ is any of $\art_{\bbQ}, \art_{\bbN}, \sshuffle$, and $\ort_{\bbN}$.}
    \label{fig:summary}
  \end{center}
\end{figure}


\section{Background}
\label{sec:background}

In this section, we will introduce the necessary background for the rest of the paper, and fix most of the notation that we will use, except
for formal definitions related to weak subsystems of second-order arithmetic, in particular $\rca$ (which consists of $\Sigma^0_1$-induction and recursive comprehension) and $\rca+\sztind$.
A standard reference for that material and, more generally, systems of interest in reverse mathematics, is \cite{simpson}.

\subsection{Generic notations}
We identify $k \in \bbN$ with the finite set $\{0, \ldots, k - 1\}$.
For every linear order $(X,<_X)$, we write $[X]^2$ for the set of pairs $(x,y)$ with $x <_X y$.
In this paper, by an \emph{interval} $I$ we always mean a pair $(u,v) \in [\bbQ]^2$, regarded as the set $]u,v[$ of rationals; we never use interval with
irrational extrema.
Finally, for any sequence $(x_n)_{n \in \bbN}$ of elements taken from
a poset, write $\limsup(x)$ for $\inf_{k \in \bbN} \sup_{n \ge k} x_n$.

\subsection{Additive and ordered colourings}\label{prelim_gen}

For the following definition, fix a linear order $(X, <_X)$.
For every poset $(P, \prec_P)$, we call a colouring $c : [X]^2 \to P$ \emph{ordered} if we have $c(x,y) \preceq_P
c(x',y')$ when $x' \le_X x <_X y \le_X y'$.
Call $c$ \emph{right-ordered} if we have $c(x,y) \preceq_P c(x, y')$ when $x <_X y \le_X y'$ (in particular being right-ordered is less restrictive than being ordered).
A colouring $c : [X]^2 \to S$ is called \emph{additive} with respect to a semigroup structure $(S, \cdot)$ if we have $c(x,z) = c(x,y) \cdot c(y,z)$
whenever $x <_X y <_X z$.
A subset $A\subseteq X$ is \emph{dense in $X$} if for every $x,y \in X$ with $x <_X y$ there is $z\in A$ such that $x<_X z<_X y$.
Given a colouring $c : [X]^n \to k$ and some interval $Y\subseteq X$, we say that $Y$ is
\emph{$c$-densely homogeneous} 
if there
exists a finite partition of $Y$ into dense subsets $D_i$ such that each
$[D_i]^n$ is monochromatic (that is, $|c([D_i]^n)| \le 1$). We will call those $c$-\emph{shuffles} if
$c$ happens to be a colouring of $\bbQ$ (i.e.\ $X=\bbQ$ and $n=1$).
Finally, given a colouring $c:\bbQ\to k$, and given an interval $I\subseteq \bbQ$,
we say that a colour $i<k$ \emph{occurs densely in $I$}
if the set of $x\in \bbQ$ such that $c(x)=i$ is dense in $I$.

\begin{definition}
    The following are statements of second-order arithmetic:
    \begin{itemize}
        \item $\weiortom_\bbQ$: for every finite poset $(P,\prec_P)$ and ordered
          colouring $c: [\bbQ]^2\to P$, there exists a $c$-homogeneous interval ${]u,v[}\subset \bbQ$.
        \item $\weishuffle$: for every $k\in\bbN$ and colouring $c:\bbQ\to k$, there exists an interval $I={]x,y[}$ such that $I$ is a $c$-shuffle. 
        \item $\art_\bbQ$: for every finite semigroup $(S,\cdot)$ and additive colouring $c:[\bbQ]^2\to S$, there exists an interval $I={]x,y[}$ such that $I$ is $c$-densely homogeneous.
        \item $\ort_\bbN$: for every finite poset $(P,\prec_P)$ and right-ordered
          colouring $c: [\bbN]^2\to P$, there exists an infinite $c$-homogeneous set.
        \item $\art_\bbN$: for every finite semigroup $(S,\cdot)$ and additive colouring $c:[\bbN]^2\to S$, there is an infinite $c$-homogeneous set.
    \end{itemize}
\end{definition}

As mentioned before, a result similar to $\art_\bbQ$ was originally proved by Shelah in \cite[Theorem 1.3 \& Conclusion 1.4]{shelah1975} and $\weishuffle$ is a central lemma
when analysing labellings of $\bbQ$ (see e.g.~\cite{CCP11}). 
We will establish that $\art_\bbQ$ and $\weishuffle$ are equivalent to $\Sigma^0_2$-induction over $\rca$ while $\ort_\bbQ$ is
provable in $\rca$.

We introduce some more terminology that will come in handy later on. Given a colouring $c:\bbQ\to k$, a set $C\subseteq k$ and an interval $I={]u,v[}$
that is a $c$-shuffle, we say that $I$ is a \emph{$c$-shuffle for the colours in $C$}, or equivalently that $I$ is \emph{$c$-homogeneous for the colours of
$C$}, if we additionally have $c(I)=C$.



\subsection{Preliminaries on Weihrauch reducibility}\label{prelim_wei}

We now give a brief introduction to the Weihrauch degrees of problems and the operations on them that we will use in the rest of the paper. We stress that here we are able to offer but a glimpse of this vast area of research, and we refer to \cite{survey-brattka-gherardi-pauly} for more details on the topic.

We deal with partial multifunctions $f : {\subseteq} \bbN^\bbN \rra \bbN^\bbN$, which we call \emph{problems}, for short. We will most often define
problems in terms of their \emph{inputs} and of the \emph{outputs} corresponding to those inputs. Elements of $\bbN^\bbN$ serve as names for the objects we are concerned with, such as colourings. Since the encoding of the objects of concern in our paper is trivial, we handle this tacitly.

A partial function $F : {\subseteq \bbN^\bbN} \to\bbN^\bbN$ is called a \emph{realizer for $f$}, which we denote by $F\vdash f$, if, for every $x\in\dom(f)$, $F(x)\in f(x)$. Given two problems $f$ and $g$, we say that $g$ is \emph{Weihrauch reducible} to $f$, and we write $g\leqW f$, if there are two computable functionals $H$ and $K$ such that $K\langle FH, \id\rangle$ is a realizer for $g$ whenever $F$ is a realizer for $f$. We define strong Weihrauch reducibility similarly: for every two problems $f$ and $g$, we say that $g$ \emph{strongly Weihrauch reduces} to $f$, written $g\leqsW f$, if there are computable functionals $H$ and $K$ such that $KFH\vdash g$ whenever $F\vdash f$.
We say that two problems $f$ and $g$ are (strongly) Weihrauch equivalent if both $f \leqW g$ and $g \leqW f$ (respectively $f \leqsW g$ and $g \leqsW f$). We write this $\equivW$ (respectively $\equivsW$).

We make use of a number of structural operations on problems, which respect the quotient to Weihrauch degrees. The first one
is the \emph{parallel product} $f \times g$, which has the power to solve an instance of $f$ and and instance of $g$ at the same time. The \emph{finite
parallelization} of a problem $f$, denoted $f^*$, has the power to solve an arbitrary finite number of instances of $f$, provided that number is given as part of
the input. Finally, the \emph{compositional product} of two problems $f$ and $g$, denoted $f \star g$, corresponds basically to the most complicated problem that can be
obtained as a composition of $f$ paired with the identity, a recursive function and $g$ paired with identity (that last bit allows to keep track of the initial input when
applying $f$).

Now let us list some of the most important\footnote{Whereas $\lpo$ and $\cn$ have been widely studied, $\tcn$ is somewhat less known (and does not appear in~\cite{survey-brattka-gherardi-pauly}): we refer to \cite{topol-comput-neumann-pauly} for an account of its properties, and to~\cite{completions-choice-brattka-gherardi} for a deeper study of some principles close to it.}
problems that we are going to use in the rest of the paper.
\begin{itemize}
    \item $\cn\colon \subseteq\baire\rra\bbN$ (\emph{closed choice on $\bbN$})
      is the problem that takes as input an enumeration $e$ of a (strict) subset of $\bbN$ and such that, for every $n\in \bbN$, $n\in \cn(e)$ if and only if $n\not\in \ran(e)$  (where $\ran(e)$ is the \emph{range} of $e$).
    \item $\tcn\colon \subseteq\baire\rra\bbN$ (\emph{totalization of closed choice on $\bbN$}) is the problem that takes as input an enumeration $e$ of \emph{any} subset of $\bbN$ (hence now we allow the possibility that $\ran (e)=\bbN$) and such that, for every $n\in\bbN$, $n\in\tcn(e)$ if and only if $n\not\in\ran(e)$ or $\ran (e)=\bbN$.
    \item $\lpo\colon 2^\bbN\to \{0,1\}$ (\emph{limited principle of omniscience}) takes as input any infinite binary string $p$ and outputs $0$ if and only if $p=0^\bbN$.
    \item $\lpo'\colon \subseteq 2^\bbN\to \{0,1\}$: takes as input (a code for) an infinite sequence $\langle p_0,p_1,\dots \rangle$ of binary strings such that the
      function $p(i) = \lim_{s\to\infty}p_i(s)$ is defined for every $i \in \bbN$, and outputs $\lpo(p)$.
\end{itemize}
Of lesser importance are the following problems:
\begin{itemize}
 \item $\Cc k$ (\emph{closed choice} on $k$) takes as input an enumeration $e$ of numbers not covering $\{0,1,\ldots,k-1\}$, and returns a number $j < k$ not covered by $e$.
 \item $\crt^1_k : k^\Nn \rra k$ (Ramsey's theorem for singletons aka the pigeon hole principle) returns some $j \in k$ on input $p \in k^\Nn$ if $j$ occurs infinitely often in $p$. We point out that $\crt^1_k \equivW (\Cc k)'\equivW \rt^1_k$ (we refer to \cite{ramsey-weihrauch-brattka-rakotoniaina,pauly-patey} for details): we prefer to use the ``colour version'' or $\rt$ for singletons since it makes many arguments more immediate than the ``set version'' would do.
 \item $\crt^1_+=\bigsqcup_{k>0} \crt^1_k$ (denoted $\rt_{1,+}$ in \cite{ramsey-weihrauch-brattka-rakotoniaina}) is the disjoint union of the $\crt^1_k$: it can be thought of as a problem taking as input a pair $(k,f)$ where $f\in\Nn$ and $f:\Nn\to k$ is a colouring, and outputting $n$ such that $f^{-1}(n)$ is infinite. By $\crt^1_\mathbb{N}$ we denote the variant where the number of colours is not provided as part of the input.
\end{itemize}
The definition of $\lpo'$ is a special case of the definition of jump as given in
\cite{survey-brattka-gherardi-pauly} applied to $\lpo$. Intuitively, $\lpo'$ corresponds to the power of answering a single binary $\Sigma^0_2$-question. In particular, $\lpo'$ is easily seen to be (strongly) Weihrauch equivalent to both $\sisfin$ and $\siscofin$, the problems accepting as input an infinite binary string $p$ and outputting $1$ if $p$ contains finitely (respectively, cofinitely) many $1$s, and $0$ otherwise.
We will use this fact throughout the paper.


\subsubsection*{The $\ect$ problem}
Another problem of combinatorial nature, introduced in \cite{comb-weak-ind-davis-hirschfeldt-hirst-pardo-pauly-yokoyama}, will prove to be very useful for the rest of the paper.

\begin{definition}
    $\ect$ is the problem whose instances are pairs $(n,f)\in \bbN\times\baire$ such that $f\colon \bbN\to n$ is a colouring of the natural numbers with $n$ colours, and such that, for every instance $(n,f)$ and $b\in \bbN$, $b\in \ect(n,f)$ if and only if
    \[
        \forall x>b\; \exists y>x \;(f(x)=f(y)).
    \]
\end{definition}
Namely, $\ect$ is the problems that, upon being given a function $f$ of the integers with finite range, outputs a $b$ such that, after that $b$, the palette of colours used is constant (hence its name, which stands for \emph{e}ventually \emph{c}onstant palette \emph{t}ail). We will refer to suitable $b$s as \emph{bounds} for the function $f$.

A very important result concerning $\ect$ and that we will use throughout the paper is its equivalence with $\tcn^*$.
\begin{lemma}[{\cite[Theorem $9$]{comb-weak-ind-davis-hirschfeldt-hirst-pardo-pauly-yokoyama}}]
\label{lemma:ectandtcn*}
    $\ect\equivW\tcn^*$
\end{lemma}

Another interesting result concerning $\ect$ is the following: if we see it as a statement of second-order arithmetic ($\ect$ can be seen as the principle asserting that for every colouring of the integers with finitely many colours there is a bound), then $\ect$ and $\sztind$ are equivalent over $\rca$ (actually, over $\rca^*$).

\begin{lemma}[{\cite[Theorem $7$]{comb-weak-ind-davis-hirschfeldt-hirst-pardo-pauly-yokoyama}}]
  \label{lem:ect-ind}
    Over $\rca$, $\ect$ and $\sztind$ are equivalent.
\end{lemma}

    Hence, thanks to the results above, it is clear why $\tcn^*$ appears as a natural candidate to be a ``translation'' of $\sztind$ in the Weihrauch degrees.


Following \cite{topol-comput-neumann-pauly}, $\isfin:2^\bbN\to\sierp$ is the following problem : for every $p\in 2^\bbN$, $\isfin(p)=\top$ if $p$ contains only finitely many occurrences of $1$ and $\isfin(p)=\bot$ otherwise
    \footnote{$\sierp$ is the Sierpinski space $\{\top,\bot\}$, where $\top$ is coded by the binary strings containing at least one $1$, and $\bot$ is coded by $0^\bbN$. $\isfin$ is strictly weaker than $\sisfin$}.

\begin{lemma}
\label{lem:isfinleqwect}
$\isfin\not\leqW\ect$
\end{lemma}
\begin{proof}
Suppose for a contradiction that a reduction exists and is witnessed by functionals $H$ and $K$. We build an instance $p$ of $\isfin$ contradicting this.

Let us consider the colouring $H(0^\bbN)$, and let $b_0\in \ect(H(0^\bbN))$ be a bound for it. Since $\isfin(0^\bbN)=\top$, the outer reduction witness will commit to answering $\top$ after having read a sufficiently long prefix of $0^\bbN$ together with $b_0$, say of length $n_0$. Now consider the colouring $H(0^{n_0}10^\bbN)$, and a bound $b_1 > b_0$ for it. Again by the fact that $\isfin(0^{n_0}10^\bbN)=\top$, there is an $n_1$ such that $K$ commits to answering $\top$ after having read the prefix $0^{n_0}10^{n_1}$ together with $b_1$. We iterate this process indefinitely and obtain an instance $p = 0^{n_0}10^{n_1}10^{n_2}1\ldots$ such that $\isfin(p) = \bot$.

However, for the colouring $H(p)$ there must be some $b_k$ which is a valid bound, as the sequence $b_0 < b_1 < b_2 < \ldots$ is unbounded. But $K$ will commit to $\top$ upon reading a sufficiently long prefix of $p$ together with $b_k$ by construction, thereby answering incorrectly.
\end{proof}

We can now assert that the two main problems that we use as benchmarks in the
sequel, namely $(\lpo')^*$ and $\tcn^*$, are incomparable in the Weihrauch lattice.

\begin{lemma}
  \label{lem:incomp-lpo-tcn}
    $(\lpo')^*$ and $\tcn^*$ are Weihrauch incomparable. Thus $(\lpo')^* \ltW (\lpo')^*\times\tcn^*$ and $\tcn^* \ltW (\lpo')^*\times\tcn^*$.
\end{lemma}

\begin{proof}
    $\tcn^*\not\leqW (\lpo')^*$: to do this, we actually show the stronger result that $\cn\not\leqW(\lpo')^*$. Suppose for a contradiction that a reduction exists, as witnessed by the computable functionals $H$ and $K$: this means that, for every instance $e$ of $\cn$, $H(e)$ is an instance of $(\lpo')^*$, and for every solution $\sigma\in(\lpo')^*(H(e))$, $K(e,\sigma)$ is a solution to $e$, i.e.\ $K(e,\sigma)\in\cn(e)$. We build an instance $e$ of $\cn$ contradicting this.

    We start by letting $e$ enumerate the empty set. At a certain stage $s$, by definition of instances of $(\lpo')^*$, $H(e\vert_s)$ converges to a certain $n$, the number of applications of $\lpo'$ that are going to be used in the reduction. Hence, however we continue the construction of $e$, there are at most $2^n$ possible values for $(\lpo')^*(H(e))$, call them $\sigma_0, \dots, \sigma_{2^n-1}$. It is now simple to diagonalize against all of them, one at a time, as we now explain. We let $e$ enumerate the empty set until, for some $s_0$ and $i_0$, $K(e\vert_{s_0},\sigma_{i_0})$ converges to a certain $m_0$: notice that such an $i_0$ has to exist, by our assumption that $H$ and $K$ witness the reduction of $\cn$ to $(\lpo')^*$. Then, we let $e$ enumerate $m_0$ at stage $s_0+1$: this implies that $\sigma_{i_0}$ cannot be a valid solution to $H(e)$, otherwise $K(e,\sigma_{{i_0}})$ would be a solution to $e$. We then keep letting $e$ enumerating $\{m_0\}$ until, for certain $s_1$ and $i_{1}$,
    $K(e\vert_{s_1},\sigma_{i_1})$ converges to $m_1$. We then let $e$ enumerate $\{m_0,m_1\}$, and continue the construction in this fashion. After $2^n$ many steps, we will have diagonalized against all the $\sigma_i$, thus reaching the desired contradiction.

    $(\lpo')^*\not\leqW\tcn^*$ is a consequence of Lemma \ref{lem:isfinleqwect}, using the fact that $\tcn^*\equivW \ect$ (see \cite{comb-weak-ind-davis-hirschfeldt-hirst-pardo-pauly-yokoyama}). To see that $\isfin\leqW\lpo'$: given any string $p\in 2^\bbN$, we consider the instance $\langle p_0,p_1,\dots\rangle$ of $\lpo'$ defined as follows: for every $i$, $p_i$ takes value $1$ until (and if) the $i$th occurrence of $1$ is found in $p$, after which point it takes value $0$. Then, $\lpo'(\langle p_0.p_1\dots\rangle)=1$ if and only if $\isfin(p)=\bot$. Hence, since $\isfin\not\leqW\ect$, we have in particular that $(\lpo')^*\not\leqW \tcn^*$.
\end{proof}

Next, we show how the problems $\crt^1_n$ and $\tcn^m$ relate. This answers a question raised by Kenneth Gill related to work in \cite{gill-phd}.

\begin{theorem}
\label{theo:rttkvstcnm}
$\crt^1_{k+1} \leqW \tcn^m$ iff $k \leq m$.
\end{theorem}

To prove the theorem, we make use of some lemmas. First, we point out that if we have $k + 1$ instances of $\tcn$, but with the promise that at least one of them is non-empty, then $k$ copies of $\tcn$ and one of $\cn$ suffice to solve the problem.

\begin{lemma}
\label{lemma:tcnotallempty}
$\tcn^{k+1}|_{\mathcal{A}(\mathbb{N}^{k+1}) \setminus \{\emptyset\}^{k+1}} \leqW \tcn^{k} \star \cn$
\begin{proof}
As we know that one of the instances $A_0,A_1,\ldots,A_k$ for $\tcn$ is non-empty, we can use $\cn$ to find some $(i,n)$ such that $n \in A_i$. We then use the $k$ copies of $\tcn$ to solve all instances $A_j$ for $j \neq i$.
\end{proof}
\end{lemma}

\begin{lemma}
\label{lemma:onestep}
If $\crt^1_{k + 1} \leqW \tcn^{m+1}$, then $\crt^1_{k} \leqW \tcn^{m}$.
\begin{proof}
Consider the input $k^\omega$ to $\crt^1_{k + 1}$. The outer reduction witness for $\crt^1_{k + 1} \leqW \tcn^{m+1}$ will have to output $k$ upon reading a sufficiently long prefix $k^\ell$ of $k^\omega$ together with some tuple $(a_0,\ldots,a_m)$ returned by $\tcn^{m+1}$. Now consider any input of the form $k^\ell p$ for $p \in \mathbf{k}^\omega$. As $k$ is a wrong answer, the inner reduction witness has to produce an instance $(A_0,\ldots,A_m)$ for $\tcn^{m+1}$ which is not $(\emptyset,\emptyset,\ldots,\emptyset)$ -- otherwise, $\tcn^{m+1}$ could output $(a_0,\ldots,a_m)$ and thereby cause the wrong answer $k$.

Together with Lemma \ref{lemma:tcnotallempty}, this gives us a reduction $\crt^1_{k} \leqW \tcn^{m} \star \cn$. Since $\crt^1_{k}$ is a closed fractal, we can drop the $\cn$ by \cite[Theorem 2.4]{paulyleroux}. This yields the claim.
\end{proof}
\end{lemma}

\begin{proof}[Proof of Theorem \ref{theo:rttkvstcnm}]
If $k > m$, we could use Lemma \ref{lemma:onestep} repeatedly to conclude that $\crt^1_2 \leqW \crt^1_{1 + k - m} \leqW \tcn^0 \equivW \id$, but this is absurd. Conversely, we have $\crt^1_{k + 1} \leqW \left ( \crt^1_2 \right )^k \leqW \tcn^k$.
\end{proof}

\subsubsection*{The finitary part of $\mathrm{WKL}'$}
In Subsection \ref{subsec:eta} we will use the notion of \emph{finitary part} of a Weihrauch degree to sidestep somewhat troublesome combinatorics. Introduced in \cite{paulycipriani1}, the $k$-finitary part of a Weihrauch degree $f$ (denoted by $\operatorname{Fin}_k(f)$) is the greatest Weihrauch degree of a problem with codomain $\mathbf{k}$ which is Weihrauch reducible to $f$ . We also write $\operatorname{Fin}(f) = \bigsqcup_{n \in \mathbb{N}} \operatorname{Fin}_n(f)$. The notion is similar to the more intensively studied first-order part of a Weihrauch degree \cite{damir-reed-keita,soldavalenti,pauly-valenti}.

Here we will calculate the $k$-finitary part of $\CCantor'$, which we view as the problem ``Given an enumeration of an infinite binary tree, find a path through it''. Another representative of $\CCantor'$ is ``Given the characteristic function of an infinite tree $T \subseteq \omega^{<\omega}$ that happens to be finitely branching, find a path through it'', which is the computational task associated to K\"onig's Lemma. It holds that $\crt^1_\mathbb{N} \leqW \CCantor'$. Thus, the $k$-finitary part of $\CCantor'$ is an upper bound for the $k$-finitary part of $\crt^1_m$, which we exploit for Corollary \ref{corr:fincrtkm} below.

\begin{theorem}
$\operatorname{Fin}_k(\CCantor') \equivW \crt^1_k$
\begin{proof}
Assume that $f : \subseteq \baire \rra \mathbf{k}$ satisfies $f \leqW \CCantor'$ via $\phi$ and $\Psi$. On input $x \in \dom(f)$, $\phi(x)$ enumerates a tree $T_x$. We produce an input $p_x$ to $\crt^1_k$ as follows: Whenever we find some $w \in T_x$ such that $\Psi(x,w) = i$, we add an $i$ to $p$. We now claim that if some $i$ appears infinitely often in $p_x$, then $i$ is a correct output for $f(x)$. Since $\Psi(x,\cdot)$ has to halt for every infinite path through $T_x$, a compactness argument shows that $\Psi(x,w)$ is defined for all but finitely many $w \in T_x$. Moreover, if $\Psi(x,w)$ is defined, then $\Psi(x,wu) = \Psi(x,w)$ for all extensions $wu$. Thus, if $\Psi(x,w) = i$ for infinitely many $w$, then by considering the infinite subtree of $T_x$ obtained by those $w$ together with the finitely $v$ such that $T(x,v)$ is undefined, we find that there is an infinite path through $T_x$ having some prefix $w$ such that $\Psi(x,w) = i$. This concludes the claim.
\end{proof}
\end{theorem}

\begin{corollary}
$\operatorname{Fin}(\CCantor') \equivW \crt^1_+$
\end{corollary}

\begin{corollary}
\label{corr:fincrtkm}
$\operatorname{Fin}_k(\crt^1_m) \equivW \crt^1_{\min\{k,m\}}$
\end{corollary}

\subsubsection*{Sequential versus parallel composition}
We use a technical result asserting that the sequential composition of $\lpo'\times\tcn$ after $\cn$ can actually be computed by the parallel product of $\lpo'$, $\tcn^n$ and $\cn$.
As customary, for every problem $\weipr P$ we write $\weipr P^n$ to mean $\underbrace{\weipr P\times\dots\times\weipr P}_{n \text{ times}}$. To see that the lemma actually applies to $\cn$, we point out that $\cn \equivW \min\cn$, where $\min\cn$ is the tightening of $\cn$ asking for the minimal valid solution.

\begin{lemma}
  \label{lem:comp-to-paral-prod}
    For all $a, b \in \mathbb{N}$ and every singlevalued problem $\weipr P : \subseteq \baire \to \baire$ with $\weipr P\leqW\cn$, it holds that $((\lpo')^a\times \tcn^b) \star \weipr P\leqW (\lpo')^a\times \tcn^b \times\weipr P$.
\end{lemma}

\begin{proof}
    The proof of $\tcn^* \equivW \ect$ in \cite[Theorem $9$]{comb-weak-ind-davis-hirschfeldt-hirst-pardo-pauly-yokoyama} actually shows that $\tcn^b \equivW \ect_{b+1}$, where $\ect_{b+1}$ is the restriction of $\ect$ to colourings with $b+1$ colours. We can thus prove the following instead:
    $$((\sisfin)^a\times \ect_b)\star\weipr P\leqW (\sisfin)^a\times \ect_b \times\weipr P$$
We observe that $\sisfin(p) = \sisfin(wp)$ for any $w \in \{0,1\}^*$, and if $n \in \ect_b(wp)$ for some $w \in \{0,1,\ldots,b-1\}^*$, then $n \in \ect(p)$. In other words, both principles have the property that adding an arbitrary prefix to an input is unproblematic. As we assume that $\weipr P \leqW \cn$, there is a finite mindchange computation that solves $\weipr P$.

In $((\sisfin)^a\times \ect_b)\star\weipr P$, we can run this finite mindchange computation to obtain the inputs for the $\operatorname{isFinite}$ and $\ect_b$-instances. Due to the irrelevance of prefixes mentioned above, the mindchanges have no problematic impact. Thus, we can actually apply $\sisfin$ and $\ect_b$ in parallel, which yields the desired reduction to $(\sisfin)^a\times \ect_b \times\weipr P$.

The singlevaluedness of $\weipr P$ makes sure that in the parallel execution we get the same solution from $\weipr P$ as the one used to compute the instances for $\sisfin$ and $\ect_b$.
\end{proof}

The following shows that the restriction to singlevalued $\weipr P$ is necessary in the statement of Lemma \ref{lem:comp-to-paral-prod}:

\begin{proposition}
$\lpo' \star \Cc 2 \nleqW \lpo' \times \Cc 2$
\begin{proof}
The problem $\lpo' \star \Cc 2$ is equivalent to ``given $p_0,p_1 \in 2^\Nn$ and non-empty $A \in \mathcal{A}(2)$, return $(i,\operatorname{isFinite}(p_i))$ for some $i \in A$.''. Let us denote this problem with $\weipr{BI}$. We will also use $\Cc 2 \times \sisfin$ instead of $\lpo' \times \Cc 2$ on the right hand side. We furthermore assume that $\mathcal{A}(2)$ is represented by $\psi : 2^\Nn \to \mathcal{A}(2)$ where $i \in \psi(p)$ iff $\exists \ell \ p(2\ell + i) = 1$.

First, we argue that $\weipr{BI} \leqW \Cc 2 \times \sisfin$ would imply $\weipr{ BI} \leqsW \Cc 2 \times \sisfin$. Let the outer reduction witness be $K : \subseteq (2^\Nn \times 2^\Nn \times 2^\Nn) \times (2 \times 2) \to (2 \times 2)$. Note that the inner reduction needs to produce inputs to $\Cc 2 \times \sisfin$ leading to all four values $(i,b) \in 2 \times 2$ -- otherwise, it would even show that $\lpo' \star \Cc 2 \leqW \lpo'$, which is known to be false for reasons of cardinality. Thus, there are prefixes $w_0,w_1$ and $0^k$ such that $K(w_0,w_1,0^k,0,0)$ converges. Restricting $\weipr{BI}$ to extensions of $w_0,w_1,0^k$ does not change its strong Weihrauch degree. We then look for extensions $w_0^1 \succ w_0$, $w_1^1 \succ w_1$ and $0^{k + \ell}$ such that $K(w_0^1,w_1^1,0^{k + \ell},0,1)$ converges, and do the same for the remaining two elements of $2 \times 2$. By restricting to extensions of those ultimate prefixes, we obtain an outer reduction witness that only depends on the $2 \times 2$-inputs, and thus witnesses a strong reduction.

Next, we disprove $\weipr{BI} \leqsW \Cc 2 \times \sisfin$. The outer reduction witness $K : 2 \times 2 \to 2 \times 2$ has to be a permutation (as all four values actually occur on the left). The inner reduction witness has to map any instance involving $0^\Nn$ as the last component to one involving $0^\Nn$ as the first component: If any prefix $(w_0,w_1,0^k)$ would lead to a $\Cc 2 \times \sisfin$-instance where the first component is not $\{0,1\}$, then by restricting to the extensions of such an input, we would obtain a reduction $\weipr{BI} \leqsW \sisfin$.

Let us consider what happens on an input $(p,p,0^\omega)$. As above, this gets mapped to some $(0^\Nn,q)$. We see that the first component of $K(i,b)$ can only depend on $b$. Moreover, as the inner reduction witness cannot map $p$s with finitely many $1$s to $q$s with infinitely many $1$s and vice versa, we actually find that the first component of $K(i,b)$ has to be $b$. Due to the symmetry of $0, 1 \in 2$, this leaves us with two candidates $K_1,K_2$ for the outer reduction witness we need to consider: $K_1(i,b) = (b,i)$ and $K_2(i,0) = (0,i)$, $K_2(i,1) = (1,1-i)$.

Next, we consider inputs of the form $(p_0,p_1,0^\Nn)$ satisfying that $\sisfin(p_0) = 1 - \sisfin(p_1)$. As above, the inner reduction witness will generate some instance $(0^\Nn,q)$. Depending on $q$, using either $K_1$ or $K_2$ as outer reduction witness yields either the answers $(0,0)$ and $(0,1)$; or the answers $(1,0)$ and $(1,1)$. However, the correct answers are either $(1,0)$ and $(0,1)$ or $(1,1)$ and $(0,0)$. Thus, both $K_1$ and $K_2$ fail, and we achieved the desired contradiction.
\end{proof}
\end{proposition}


\subsection{Green theory}

Green theory is concerned with analysing the structure of ideals of finite
semigroups, be they one-sided on the left or right or even two-sided.
This gives rise to a rich structure to otherwise rather inscrutable algebraic
properties of finite semigroups.
We will need only a few related results, all of them relying on
the definition of the \emph{Green preorders} and of idempotents
(recall that an element $s$ of a semigroup is idempotent when $ss = s$).

\begin{definition}
\label{def:green}
  For a semigroup $(S, \cdot)$, define the \emph{Green preorders} as follows:
  \vspace{-0.5em}
  \[
    \begin{array}{c l c l r}
 \quad\bullet \;\; &  s \le_{\GrR} t & \text{  if and only if  } &  s=t \text{ or } s \in t  S=\{t  a ~\colon a\in S\} & \scriptsize\text{(suffix order)} \\
 \quad\bullet \;\; &  s \le_{\GrL} t & \text{  if and only if  } & s=t \text{ or } s \in S  t=\{a  t ~\colon a\in S\} & \scriptsize\text{(prefix order)} \\
 \quad\bullet \;\; &  s \le_{\GrH} t & \text{  if and only if  } & s \le_{\GrR} t \text{ and }s \le_{\GrL} t & \\
 \quad\bullet \;\; &  s \le_{\GrJ} t & \text{  if and only if  } &
                       \multicolumn{2}{l}{
                       s \le_{\GrR} t \text{ or } s \le_{\GrL} t \text{ or }
                       s \in S t  S = \{ a  t  b ~ \colon (a,b) \in S^2\}
                                                                   } \\
      & & & & \scriptsize{\text{(infix order)}}

    \end{array}
\vspace{-0.5em}
\]
The associated equivalence relations are written $\GrR$, $\GrL$, $\GrH$, $\GrJ$; their equivalence classes are called respectively $\GrR$, $\GrL$, $\GrH$, and $\GrJ$-classes.
\end{definition}




We conclude this section reporting, without proof, three technical lemmas that will be needed in \Cref{sec:artq} and \ref{sec:artn}. Although not proved in second-order
arithmetic originally, it is clear that their proofs go through in $\rca$:
besides straightforward algebraic manipulations, they only rely on the existence,
for each finite semigroup $(S, \cdot)$, of an index $n \in \bbN$ such that $s^n$ is idempotent for any $s \in S$.

\begin{lemma}[{\cite[Proposition~A.2.4]{PinPerrin}}]
\label{lem:Hidemgroup}
If $(S, \cdot)$ is a finite semigroup, $H\subseteq S$ an $\GrH$-class, and some $a,b\in H$ satisfy $a \cdot b\in H$ then for some $e\in H$ we know that $(H, \cdot, e)$ is a group.
\end{lemma}

\begin{lemma}[{\cite[Corollary~A.2.6]{PinPerrin}}]
  \label{lem:leRJR}
For any pair of elements $x, y \in S$ of a finite semigroup, if we have $x \le_{\GrR} y$ and $x, y$ $\GrJ$-equivalent,
then $x$ and $y$ are also $\GrR$-equivalent.
\end{lemma}

\begin{lemma}[{\cite[Corollary~A.2.6]{PinPerrin}}]
\label{lem:leLRH}
For every finite semigroup $S$ and $s, t \in S$, $s \le_{\GrL} t$ and $s \GrR t$ implies $s \GrH t$.
\end{lemma}


\section{The shuffle principle and related problems}
\subsection{The shuffle principle in reverse mathematics}
\label{subsec:revmath-shuffle}
We start by giving a proof\footnote{From Leszek A. Ko{\l}odziejczyk, personal
communication.} of the shuffle principle in $\rca+\sztind$, since, in a way, it gives a clearer picture of some properties of shuffles that we use in the rest of the paper.

\begin{lemma}
  \label{lem:sztind-shuffle}
    $\rca+\sztind\vdash\weishuffle$
\end{lemma}
\begin{proof}
    Let $c: \bbQ \to n$ be a colouring of the rationals with $n$ colours.
For any natural number $k$, consider the following $\Sigma^0_2$ formula $\varphi(k)$:
``there exists a finite set $L \subseteq n$ of cardinality $k$ and there exist $u, v \in \bbQ$ with $u < v$ such that $c(w) \in L$ for every $w \in \ooint{u,v}$''. 
Since $\varphi(n)$ is true, it follows from the $\Sigma^0_2$ minimization principle that there exists a minimal $k$ such that $\varphi(k)$ holds. Consider $u,v \in \bbQ$ and the set of colours $L$ corresponding to this minimal $k$. We now only need to show that
$\ooint{u,v}$ is a $c$-shuffle to conclude.

Let $a = c(x)$ for some $x \in \ooint{u,v}$. We need to prove that $a$ occurs densely in $\ooint{u,v}$. Consider arbitrary $x, y \in \ooint{u,v}$ with $x < y$. We are done if we show that there exists some $w \in \ooint{x,y}$ with $c(w) = a$.
So, suppose that there is no such $w$. Consider the set of colours $L\setminus \{a\}$
Clearly, $\varphi(\card{L\setminus \{a\}})$ holds. However, $\card{L\setminus \{a\}}<\card{L}=k$,
contradicting the choice of $k$ as the minimal number such that $\varphi(k)$ holds.

\end{proof}

The proof above shows an important feature of shuffles: given a certain interval $]u,v[$, any of its subintervals having the fewest colours is a shuffle.
Interestingly, the above implication reverses, so we have the following equivalence.

\begin{theorem}
\label{thm:revmath-shuffle}
     Over $\rca$, $\weishuffle$ is equivalent to $\sztind$.
\end{theorem}
We do not offer a proof of the reversal here; such a proof can easily be done
by taking inspiration from the argument we give for \Cref{lem:ect-leq-ishuffle}.

With this equivalence in mind, we now introduce Weihrauch problems corresponding to $\weishuffle$,
beginning with the stronger one.

\begin{definition}
\label{def:wei-shuffle}
  We regard $\sshuffle$ as the problem with instances $(k,c)\in \bbN\times\baire$ such that $c:\bbQ\to k$ is a colouring of the rationals with $k$ colours, and such that, for every instance $(k,c)$, for every pair $(u,v)\in [\bbQ]^2$ and for every $C\subseteq k$, $(u,v,C)\in\sshuffle(k,c)$ if and only if $]u,v[$ is a $c$-shuffle for the colours in $C$.
\end{definition}

Note that the output of $\sshuffle$ contains two components that cannot be easily computed from one another.
It is very natural to split the principles into several problems, depending on the type of solution that we want to be given: one problem will output the \emph{colours} of a shuffle, whereas another will output the \emph{interval}. As we will see, the strength of these two versions of the same principle have very different uniform strength.

\begin{definition}
\label{def:wei-weak-shuffle}
  $\ishuffle$ (``i'' for ``interval'') is the same problem as $\sshuffle$ save for the fact that a valid output only contains the interval $]u,v[$ which is a $c$-shuffle.
  Complementarily, $\cshuffle$ (``c'' for ``colour'') is the problem that only outputs a possible set of colours taken by a $c$-shuffle.
\end{definition}

We will first start analysing the weaker problems $\cshuffle$ and $\ishuffle$
and show they are respectively equivalent to $(\lpo')^*$ and $\tcn^*$.
This will also imply that $\sshuffle$ is stronger than $(\lpo')^* \times \tcn^*$,
but the converse will require an entirely distinct proof.

Our definitions include the number of colours as part of the input. We discuss in Section \ref{sec:colours} how the variants with an unknown but still finite number of colours relate to the versions we focus on.

\subsection{Weihrauch complexity of the weaker shuffle problems}

We first provide a classification of $\cshuffle$, by gathering a few lemmas. The first also applies to $\ishuffle$ and $\sshuffle$.

\begin{lemma}\label{lem:cshuffle-finparall}
    $\cshuffle^*\equivW\cshuffle$, $\ishuffle^*\equivW\ishuffle$ and $\sshuffle^*\equivW\sshuffle$.
\end{lemma}
\begin{proof}
  Without loss of generality, let us show that we have $\cshuffle\times \cshuffle \leqW \cshuffle$, $\ishuffle \times \ishuffle \leqW \ishuffle$ and $\sshuffle\times \sshuffle\leqW \sshuffle$.
  Consider the pairing of the two input colourings. To give more details, let $(n_0,f_0)$ and $(n_1,f_1)$ be instances of $\sshuffle$. Let us fix a computable bijection $\langle \cdot,\cdot\rangle:n_0\times n_1\to n_0n_1$ and
    define the colouring $f\colon\bbQ\to n_0n_1$ by $f(x)=\langle f_0(x),f_1(x)\rangle$ for every $x\in\bbQ$. Hence, $(n_0n_1,f)$ is a valid instance of $\sshuffle$.

    Let
    $C\in\cshuffle(n_0n_1,f)$: this means that there is an interval $I$ that is a $f$-shuffle for the colours of $C$. For $i<2$, let $C_i:=\{ j: \exists
    c\in C( j=\pi_i(j)) \}$, where $\pi_i$ is the projection on the $i$th component. Then, $C_i\in\cshuffle(n_if_i)$, as witnessed by the interval $I$.

    With the same reasoning, if $I \in \ishuffle(n_0n_1,f)$, then also $I \in \ishuffle(n_0,f_0)$ and $I \in \ishuffle(n_1,f_1)$. Finally, if $(I,C) \in \sshuffle(n_0n_1,f)$, then $(I,C_0) \in \sshuffle(n_0,f_0)$ and $(I,C_1) \in \sshuffle(n_1,f_1)$.

    To conclude $\sshuffle^*\equivW\sshuffle$ from $\sshuffle\times \sshuffle \leqW \sshuffle$, we just need to observe that $\sshuffle$ has a computable instance; likewise for $\cshuffle$ and $\ishuffle$.

\end{proof}

Let $\cshuffle_n$ be the restriction of $\cshuffle$ to $n$-colourings. The following lemma is due to a suggestion by an anonymous referee:

\begin{lemma}
\label{lem:lpo-leq-cshuffle}
$(\lpo')^n \leqW \cshuffle_{n+1}$, but $(\lpo')^n \nleqW \cshuffle_{n}$.
\end{lemma}
\begin{proof}
  We will prove that $(\sisfin)^n \leqsW\cshuffle_{n+1}$. Given an input $( p_0,p_1,\ldots,p_{n-1})$ for $(\sisfin)^n$ we define a colouring $c : \mathbb{Q} \to (n+1)$ of the rationals by setting $c(\frac{k}{nm + \ell})= \ell$ if $p_{\ell} (m) = 1$ and $c(\frac{k}{nm + \ell})= n$ if $p_{\ell} (m) = 0$, where $\operatorname{gcd}(k,nm+ \ell) = 1$ and $\ell < n$. Let $C$ be a result from applying $\cshuffle_{n+1}$. We find that $\ell \in C$ iff $p_\ell$ contains infinitely many $1$s -- as a colour $\ell < n$ appears either only finitely many times, or is dense in every interval. The colour $n$ is only used as neutral, it is otherwise disregarded.

To see that $(\lpo')^n \nleqW \cshuffle_{n}$ we appeal to the cardinality of the possible outputs. For $(\lpo')^n$ there are $2^n$ different outputs, and each finite prefix can be extended to instances having any one of them as the only correct answer. On the other hand, $\cshuffle_{n}$ has only non-empty subsets of $n$ as possible answers, i.e.~there are only $2^n - 1$ possible answers.
\end{proof}

We can also prove that finitely many copies of $\lpo'$ suffice to solve $\cshuffle_n$. However, our construction involves an exponential increase in the parameter. We do not know whether this is necessary.

\begin{lemma}
\label{lem:cshuffle-leq-lpo}
$\cshuffle_n\leqW (\lpo')^{2^n-2}$
\end{lemma}
\begin{proof}We actually show that $\cshuffle_n\leqW\sisfin^{2^n-1}$. Let $(n,c)$ be an instance of $\cshuffle$. 
  The idea is that we will use one instance of $\sisfin$ for every non-empty strict subset $C$ of the set of colours $n$, in order to determine for which such $C$s there exists an interval $I_C$ such that $c(I_C) \subseteq C$. We will then prove that any $\subseteq$-minimal such $C$ is a solution for $(n,c)$.

  Let $C_i$, for $i<2^n-2$, be an enumeration of the non-empty strict subsets of $n$.
    Let $I_j$ be an enumeration of the open intervals of $\bbQ$, and let $q_h$ be an enumeration of $\bbQ$. For every $i<2^n-2$, we build an instance $p_i$ of $\sisfin$ in stages in parallel. At every stage $s$, for every component $i<2^n-2$, there will be a ``current interval'' $I_{j_s^i}$ and a ``current point'' $q_{h_s^i}$. We start the construction by setting the current interval to $I_0$ and the current point to $q_0$ for every component $i$.

    For every component $i$, at stage $s$ we do the following:
    \begin{itemize}
        \item if $q_{h_s^i}\not\in I_{j_s^i}$ or if $c(q_{h_s^i})\in C_i$, we set $I_{j_{s+1}^i}=I_{j_{s}^i}$ and $q_{h_{s+1}^i}=q_{h_{s}+1}$. Moreover, we set $p_i(s)=0$. In practice, this means that if the colour of the current point is in $C_i$, or if the current point is not in the current interval, no special action is required, and we can move to consider the next point.
        \item If instead $q_{h_s^i}\in I_{j_s^i}$ and $c(q_{h_s^i})\not\in C_i$, we set $I_{j_{s+1}^i}=I_{j_s+1}$ and $q_{h_{s+1}^i}=q_0$. Moreover, we set $p_i(s)=1$. In practice, this means that if the current point is in the current interval and its colour is not a colour of $C_i$, then, we need to move to consider the next interval in the list, and therefore we reset the current point to the first point in the enumeration. Moreover, we record this event by letting $p_i(s)$ take value $1$.
    \end{itemize}
    We iterate the construction for every $s\in\bbN$. After infinitely many steps, we obtain an instance $\langle p_0,p_1,\dots, p_{2^n-3}\rangle$ of
    $\sisfin^{2^n-2}$.

    Let $C$ be a $C_i$ be such that $\sisfin(p_i) = 1$ but $\sisfin(p_j) = 0$ for all $j$ such that $C_j \subset C_i$; or $C = \{0,1,\ldots,n-1\}$ if no such $i$ exists. We claim that $C$ is a correct answer for $\cshuffle_n$. If $C = C_i$, then the construction for $C_i$ only proceeded to a different interval finitely many times. The last interval $I$ we consider for $C_i$ has no colour in it not appearing in $C_i$. Moreover, we know that for all $C_j \subset C_i$ the process never stabilized, which means that every subinterval of $I$ has all colours in $C_i$ appearing. If $C = \{0,1,\ldots,n-1\}$, then every colour is dense everywhere, and thus $C$ is the only correct answer to $\cshuffle_n$.
\end{proof}

Putting the previous lemmas together, we have the following:
\begin{theorem}\label{thm:lpo-cshuffle}
    $(\lpo')^*\equivW \cshuffle$
\end{theorem}
\begin{proof}
  $(\lpo')^*\leqW\cshuffle$ is given by Lemmas \ref{lem:cshuffle-finparall} and \ref{lem:lpo-leq-cshuffle}. For the other direction, notice that $\cshuffle\equivW \bigsqcup_{n\in\bbN}\cshuffle_n$. The result then follows from \Cref{lem:cshuffle-leq-lpo}.
\end{proof}

While Theorem \ref{thm:lpo-cshuffle} tells us that for any finite number of parallel $\lpo'$-instances can be reduced to $\cshuffle$ for $m$-colourings for a suitable choice of $m$, and vice versa, a sufficiently large number of $\lpo'$-instances can solve $\cshuffle$ for $m$-colourings, one direction of our proof involved an exponential increase in the parameter. Before moving on to $\ishuffle$, we thus raise the open question of whether this gap can be narrowed:

\begin{question}
What is the relationship between $(\lpo')^n$ and $\cshuffle_m$ for individual $n,m \in \mathbb{N}$?
\end{question}

\begin{lemma}
  \label{lem:ishuffle-leq-tcn}
    Let $\ishuffle_n$ be the restriction of $\ishuffle$ to the instances of the form $(n,c)$.
    For $n \geq 1$, it holds that $\ishuffle_{n} \leqsW\tcn^{n-1}$.
\end{lemma}
\begin{proof}
    Fix an enumeration $I_j$ of the intervals of $\bbQ$, an enumeration $q_h$ of $\bbQ$,
    a computable bijection $\langle \cdot,\cdot\rangle\colon \bbN\times\bbN\to\bbN$,
    and let $(n,c)$ be an instance of $\ishuffle_n$.

  The idea of the reduction is the following: with the first instance $e_{n-1}$ of $\tcn$, we look for an interval $I_j$ on which $c$ takes only $n-1$ colours: if no such interval exists, then this means that every colour is dense in every interval, and so every $I_j$ is a valid solution to $c$. Hence, we can suppose that such an interval is eventually found: we will then use the second instance $e_{n-2}$ of $\tcn$ to look for a subinterval of $I_j$ where $c$ takes only $n-2$ values. Again, we can suppose that such an interval is found. We proceed like this for $n-1$ steps, so that in the end the last instance $e_1$ of $\tcn$ is used to find an interval $I'$ inside an interval $I$ on which we know that at most two colours appear: again, we look for $c$-monochromatic intervals: if we do not find any, then $I'$ is already a $c$-shuffle, whereas if we do find one, then that interval is now a solution to $c$, since $c$-monochromatic intervals are trivially $c$-shuffles..

    Although not apparent in the sketch given above, an important part of the proof is that the $n-1$ searches we described can be performed \emph{in parallel}: the fact that this can be accomplished relies on the fact that any subinterval of a shuffle is a shuffle.  More formally, we proceed as follows: we define $n-1$ instances $e_1,\dots,e_{n-1}$ of $\tcn$ as follows. For every stage $s$, every instance $e_i$ will have a ``current interval'' $I_{j_s^i}$ and a ``current point'' $q_{h_s^i}$ and a ``current list of colours'' $L_{k_s^i}$. We start the construction by the setting the current interval equal to $I_0$, the current point equal to $q_0$ and the current list of points equal to $\emptyset$ for every $i$.

    At stage $s$, there are two cases:
    \begin{itemize}
    \item if, for every $i$, $q_{h_s^i}\not\in I_{j_s^i}$ or $| L_{k_s^i} \cup \{ c(q_{h_s^i})\}|\leq i$, we set $I_{j_{s+1}^i}=I_{j_s^i}$, $q_{h_{s+1}^i}=q_{h_s^i+1}$ and $L_{k_{s+1}^i}=L_{k_s^i}\cup \{ c(q_{h_s^i})\}$. Moreover, we let every $e_i$ enumerate every number of the form $\langle s, a\rangle$, for every $a\in \bbN$, except for $\langle s, j^i_s\rangle$. We then move to stage $s+1$.

      In practice, this means that if the set of colours of the points of the current interval seen so far does not have cardinality larger than $i$, no particular action is required, and we can move to check the next point on the list.
        \item otherwise: let $i'$ be maximal such that $q_{h_s^i}\in I_{j_s^i}$ and $|L_{k_s^i} \cup \{ c(q_{h_s^i})\}|> i$. Then, for every $i> i'$ we proceed as in the previous case (i.e., the current interval, current point, current list of colours and enumeration are defined as above). For the other components, we proceed as follows: we first look for the minimal $\ell>j_s^{i'}$ such that $I_\ell\subseteq I_{j_s^{i'+1}}$ (if $i'=n-1$, just pick $\ell=j_s^{n-1}+1$). 
          Then, for every $i\leq i'$, we set $I_{j_{s+1}^i}=I_\ell$, $q_{h_{s+1}^i}=q_0$ and $L_{k_{s+1}^i}=\emptyset$. Moreover, we let $e_i$ enumerate every number of the form $\langle t, a\rangle$ with $t<s$ that had not been enumerated so far, and also every number of the form $\langle s, a\rangle$, with the exception of $\langle s, j_s^i\rangle$. We then move to stage $s+1$.

          In practice, this means that if, for a certain component $i'$, we found that the current interval has too many colours, then, for all the components $i\leq i'$, we move to consider intervals strictly contained in the current interval of component $i'$.
    \end{itemize}
    We iterate the procedure for every $s\in\bbN$, thus obtaining the $\tcn^{n-1}$-instance $\langle e_1,\dots,e_{n-1}\rangle$.

    Let $\sigma \in \bbN^{n-1}$ be such that $\sigma\in \tcn^{n-1}(\langle e_1,\dots,e_{n-1}\rangle)$. Then, we look for the minimal $i$ such that $I_{\pi_2(\sigma(i))}\subseteq I_{\pi_2(\sigma(i+1))}\subseteq\dots\subseteq I_{\pi_2(\sigma(n-1))}$ (by $\pi_i$ we denote the projection on the $i$th component, so $\langle \pi_1(x),\pi_2(x)\rangle=x$)). We claim that $I_{\pi_2(\sigma(i))}$ is a $c$-shuffle, which is sufficient to conclude that $\ishuffle_n\leqsW\tcn^{n-1}$.

    We now prove the claim. First, suppose that $e_{n-1}$ enumerates all of $\bbN$. Then, the second case of the construction was triggered infinitely many times with $i'=n-1$: hence, no interval contains just $n-1$ colours, and so, as we said at the start of the proof, this means that every interval is a $c$-shuffle. In particular, this imples that $I_{\pi_2(\sigma(i))}$ is a valid solution. Hence we can suppose that $e_{n-1}$ does not enumerate all of $\bbN$.

    Next, we notice that for every $m>1$, if $e_m$ enumerates all of $\bbN$, the so does $e_{m-1}$, by inspecting the second case of the construction. Let $m$ be minimal such that $e_m$ does not enumerate all of $\bbN$. For such an $m$, it is easy to see that $I_{\pi_2(\sigma(m))}$ is a valid solution to $c$: indeed, we know from the construction that $c$ takes $m$ colours on $I_{pi_2(\sigma(m))}$, and that for no interval contained in $I_{\pi_2(\sigma(m))}$ $c$ takes $m-1$ colours, which means that $I_{\pi_2(\sigma(m))}$ is a $c$-shuffle. Moreover, it is easy to see that $I_{\pi_2(\sigma(m))}\subseteq I_{\pi_2(\sigma(m+1))}\subseteq\dots I_{\pi_2(\sigma(n-1))}$, which implies that $i\leq m$. Since every subinterval of a $c$-shuffle is a $c$-shuffle, $I_{\pi_2(\sigma(i))}$ is a valid solution to $c$, as we wanted.
\end{proof}

\begin{lemma}\label{lem:ect-leq-ishuffle}
Let $\ect_n$ be the restriction of $\ect$ to the instances of the form $(n,f)$. It holds that $\ect_n\leqsW \ishuffle_n$.
\end{lemma}
\begin{proof}
    Let $(n,f)$ be an instance of $\ect_n$. We define $c\colon \mathbb{Q}\to n$ by $c(\frac{a}{b})=f(b)$ where $\operatorname{gcd}(a,b) = 1$. Hence, all the points of the same denominator have the same colour according to $c$. Let $(\frac{u}{k},\frac{v}{\ell})\in \ishuffle_n(n,c)$. Let $b$ be such that $\frac{1}{b}< \frac{v}{\ell} -  \frac{u}{k}$. We claim that $b$ is a bound for $f$. Suppose not, then there is a colour $i<n$ and a number $x\in\bbN$ such that $x>b$ and $f(x)=i$, but for no $y>x$ it holds that $f(y)=i$. Hence, all the reduced of the form $\frac{w}{x}$ are given colour $i$, but $i$ does not appear densely often in any interval of $\mathbb{Q}$. But by choice of $b$, there is a $z\in\mathbf{Z}$ such that $\frac{z}{b}\in \left] \frac{u}{k},\frac{v}{\ell} \right[$, which is a contradiction. Hence $b$ is a bound for $f$.

\end{proof}

Putting things together, we finally have a characterization of $\ishuffle$. We even get a precise characterization for each fixed number of colours.

\begin{theorem}\label{thm:tcn-ect-ishuffle}
  We have the Weihrauch equivalence
  \vspace{-0.5em}
  \[\ect_n\equivW\ishuffle_n\equivW \tcn^{n-1} \qquad \text{whence} \qquad
  \ect\equivW \ishuffle\equivW \tcn^* \]
\end{theorem}
\begin{proof}
We get $\tcn^{n-1}\leqW\ect_n$ by inspecting the second half of \cite[Theorem 9]{comb-weak-ind-davis-hirschfeldt-hirst-pardo-pauly-yokoyama}. Then Lemma \ref{lem:ect-leq-ishuffle} gives us $\ect_n \leqW \ishuffle_n$. Lemma \ref{lem:ishuffle-leq-tcn} closes the cycle by establishing $\ishuffle_n\equivW \tcn^{n-1}$.
\end{proof}

\subsection{The full shuffle problem}

The main result of this section is that $\sshuffle\equivW \tcn^*\times(\lpo')^*$, which will be proved in \Cref{thm:sshuffle}. For one direction, we merely need to combine our results for the weaker versions:

\begin{lemma}
    $\tcn^*\times(\lpo')^*\leqW \sshuffle$
\end{lemma}
\begin{proof}
  From \Cref{thm:lpo-cshuffle} and \Cref{thm:tcn-ect-ishuffle}, we have that $\tcn^*\times(\lpo')^*\leqW \ishuffle\times\cshuffle$, and since clearly $\ishuffle \leqW\sshuffle$ and $\cshuffle\leqW\sshuffle$, by \Cref{lem:cshuffle-finparall} we have that $\tcn^*\times(\lpo')^*\leqW\sshuffle$.
\end{proof}

For the other direction, again, we want to be precise as to the number of $\tcn$- and $(\lpo')$-instances we use to solve an instance of $\sshuffle$. Note that we will use a far larger number of $\tcn$-instances to obtain a suitable interval than we used in \Cref{lem:ishuffle-leq-tcn}.

\begin{lemma}
\label{lem:sshuffle-leq-tcnlpo}
    Let $\sshuffle_n$ be the restriction of $\sshuffle$ to the instances of the form $(n,c)$. Then, $\sshuffle_n \leqW (\tcn\times\lpo')^{2^n-1}$
\end{lemma}
\begin{proof}
  Let $(n,c)$ be an instance of $\sshuffle$. The idea of the proof of $\sshuffle_n\leqW (\tcn \times \lpo')^{2^n-1}$ is, in essence, to combine the proofs of \Cref{lem:ishuffle-leq-tcn} and of \Cref{lem:cshuffle-leq-lpo}: we want to use $\tcn$ to find a candidate interval for a certain subset $C$ of $n$, and on the side we use $\lpo'$ (or equivalently, $\sisfin$) to check for every such set $C$ whether a $c$-shuffle for the colours of $C$ actually exists.
    The main difficulty with the idea described above is that the two proofs must be intertwined, in order to be able to find both a $c$-shuffle and the set of colours that appears on it.

    We proceed as follows: let $C_i$ be an enumeration of the non-empty subsets of $n$. Moreover, let us fix some computable enumeration $I_j$ of the intervals of $\bbQ$, some computable enumeration $q_h$ of the points of $\bbQ$, and some computable bijection $\langle \cdot, \cdot\rangle\colon \bbN\times\bbN\to\bbN$. For every $C_i$, we will define an instance $\langle p_i,e_i \rangle$ of $\sisfin\times \tcn$ in stages as follows: at every stage $s$, for every index $i$, there will be a ``current interval'' $I_{j_s^i}$ and a ``current point'' $q_{h_s^i}$. We begin stage $0$ by setting the current interval to $I_0$ and the current point to $q_0$ for every index $i$.

    At stage $s$, for every component $i$, there are two cases:
    \begin{itemize}
    \item if $q_{h_s^i}\not\in I_{j_s^i}$ or if $c(q_{h_s^i})\in C_i$, we set $I_{j_{s+1}^i}=I_{j_{s}^i}$ and $q_{h_{s+1}^i}=q_{h_{s}^i+1}$. Moreover, we set $p_i(s)=0$ and we let $e_i$ enumerate all the integers of the form $\langle s, a\rangle$, except $\langle s, j^i_{s+1}\rangle$. We then move to stage $s+1$.

      In plain words, for every component $i$, we check if the colour of the current point is in $C_i$, or if the current point is not in the current interval: if this happens, no special action is required.
    \item If instead $q_{h_s^i}\in I_{j_s^i}$ and $c(q_{h_s^i})\not\in C_i$, we set $I_{j_{s+1}^i}=I_{j_s^i+1}$ and $q_{h_{s+1}^i}=q_0$. Moreover, we set $p_i(s)=1$, and we let $e_i$ enumerate all the numbers of the form $\langle t,a\rangle$, for $t<s$, that had not been enumerated at a previous stage, and also all the numbers of the form $\langle s, a\rangle$, with the exception of $\langle s,j^i_{s+1}\rangle$. We then move to stage $s+1$.

      In plain words: if we find that for some component $i$ the colour of the current point is not in $C_i$, then, from the next stage, we start considering another interval, namely the next one in the fixed enumeration. We then reset the current point to $q_0$ (so that all rationals are checked again), and we record the event by letting $p_i(s)=1$ and changing the form of the points that $e_i$ is enumerating.
    \end{itemize}
    We iterate the procedure for every integer $s$.
    Let $\sigma\in (2\times \bbN)^{2^n-1}$ be such that
    \[\sigma\in (\sisfin\times \tcn)^{2^n-1}(\langle \langle p_1,e_1\rangle \dots,\langle p_{2^n-1},e_{2^n-1}\rangle)\rangle\]
    Let $k$ be the minimal cardinality of a subset $C_i\subseteq n$ such that $\sisfin (p_i)=1$: 
    notice that such a $k$ must exist, because $c$-shuffle exist. Then, we claim that the corresponding $I_{\pi_2(\sigma(i))}$ is a $c$-shuffle (by $\pi_i$ we denote the projection on the $i$th component, so $\langle \pi_1(x),\pi_2(x)\rangle=x$)). If we do this, it immediately follows that $\sshuffle\leqW ((\lpo')\times \tcn)^{2^n-1}$.

    Hence, all that is left to be done is to prove the claim. By the fact that $\sisfin(p_i)=1$, we know that the second case of the construction is triggered only finitely many times. Hence, $e_i$ does not enumerate all of $\bbN$, and so $I_{\pi_2(\sigma(i))}$ is an interval containing only colours from $C_i$. Moreover, by the minimality of $|C_i|$, we know that no subinterval of $I_{\pi_2(\sigma(i))}$ contains fewer colours, which proves that $I_{\pi_2(\sigma(i))}$ is a $c$-shuffle.
 \end{proof}

Putting the previous results together, we obtain the following.
\begin{theorem}\label{thm:sshuffle}
   $\sshuffle\equivW \tcn^*\times(\lpo')^*$
\end{theorem}



\subsection{The $\etap$-problem}
\label{subsec:eta}
A weakening of the shuffle principle was studied in \cite{FrittaionP17} under the name $\etap$. The principle $\etap$ asserts that for any colouring of $\mathbb{Q}$ in finitely many colours, some colour will be dense somewhere. We formalize it here as follows:

\begin{definition}
The principle $\etap$ takes as input a pair $(k,\alpha)$ where $k\in\Nn$ and $\alpha:\mathbb{Q}\to k$ is a colouring, and returns an interval $I$ and a colour $n<k$ such that $\alpha^{-1}(n)$ is dense in $I$. The principle $\ietap$ returns only the interval $I$, $\cetap$ only the dense colour $n$. Let $(\cetap)_k$ be the restriction of $\cetap$ to $k$-colourings.
\end{definition}

An important aspect of the definition above to notice is that we require a bound on the number of colours used to be declared in the instance of $\etap$. We discuss what happens if the number of colours is unknown and merely promised to be finite in Section \ref{sec:colours}.

While $\etap$ also exhibits the pattern that we can neither compute a suitable interval from knowing the dense colour nor vice versa, we shall see that as far as the Weihrauch degree is concerned, finding the interval is as hard as finding both interval and colour. Our proof does not preserve the number of colours though - we leave open whether this can be avoided.

\begin{proposition}
$\etap \equivW \ietap \equivW \tcn^* \equivW \ishuffle$
\begin{proof}
Taking into account Theorem \ref{thm:tcn-ect-ishuffle}, it suffices for us to show that $\etap \leqW \ishuffle$ and that $\ect \leqW \ietap$. For $\etap \leqW \ishuffle$ we observe that an interval which is a shuffle not only has a dense colour in it, but every colour that appears is dense. Thus, we return the interval obtained from $\ishuffle$ on the same colouring, together with the first colour we spot in that interval.

It remains to prove that $\ect \leqW \ietap$. Given a $k$-colouring $c$ of $\mathbb{N}$, we will compute a $2^k$-colouring $\alpha$ of $\mathbb{Q}$. We
  view the $2^k$-colouring as a colouring by subsets of $k$, i.e.~each rational gets assigned a set of the original colours. To determine whether the
  $n$-th rational $q_n$ should be assigned the colour $j < k$, we consider the number $m_{n,j} = |\{s \mid s \leq n \wedge c(s) = j\}|$ of prior ocurrences
  of the colour $j$ in $c$. If the integer part of $q_n \cdot 2^{m_{n,j}}$ is odd, $q_n$ is assigned colour $j$, otherwise not.

If $j$ appears only finitely many times in $c$, then $m_{n,j}$ is eventually constant, and the distribution of $j$ in $\alpha$ follows (with finitely many exceptions) the pattern of alternating intervals of width $2^{-m_{n,j}}$. This ensures that none of the $2^k$-many colours for $\alpha$ can be dense on an interval wider than $2^{-m_{n,j}}$. Subsequently, we find that the width of the interval having a dense colour returned by $i\etap$ provide a suitable bound to return for $\ect$.
\end{proof}
\end{proposition}

\begin{question}
How do the restrictions $(\etap)_k$ and $(\ietap)_k$ to $k$-colourings compare to $\tcn^\ell$?
\end{question}

The proposition above implies that $\cetap$ has to be weaker than $(\lpo')^*$, since it is immediate to see that it is computed by both $\etap$ and $\cshuffle$. We now give more bounds on its strength.

\begin{lemma}
$(\cetap)_{k+1} \leqW \crt^1_{k+1} \times (\cetap)_k$
\begin{proof}
Fix some enumeration $(I_n)_{n \in \mathbb{N}}$ of all rational intervals. The forwards reduction witness is constructed as follows. We keep track of an interval index $n$ and a colour $c$, starting with $n = 0$ and $c = 0$. We keep writing the current value of $c$ to the input of $\crt^1_{k+1}$, and we construct a colouring $\beta : \mathbb{Q} \to \{0,1,\dots,k-1\}$ by scaling the colouring $\alpha$ restricting to $I_n$ up to $\mathbb{Q}$, while excluding $c$ and subtracting $1$ from every colour $d > c$. The fact that we may have already assigned $\beta$-colours to finitely many points in a different way before is immaterial.

If we ever find a rational $q \in I_n$ with $\alpha(q) = c < k$, we increment $c$. If we find $q \in I_n$ with $\alpha(q) = c = k$, we set $c = 0$ and increment $n$. In particular, we stick with any particular $I_n$ until we have found points of all different colours inside it.

The backwards reduction witness receives two colours, $c \in \{0,1,\ldots,k\}$ and $d \in \{0,1,\ldots,k-1\}$. If $d < c$, it returns $d$. If $d \geq c$, it returns $d + 1$.

To see that the reduction works correctly, first consider the case where every colour is dense everywhere. In this case, everything is a correct answer, and the reduction is trivially correct. Otherwise, there has to be some interval $I_n$ and some colour $c$ such that $\alpha^{-1}(c) \cap I_n = \emptyset$. In this case, our updating of $n$ and $c$ will eventually stabilize at such a pair. The answer we will receive from $\crt^1_{k+1}$ is $c$. Apart from finitely many points, $\beta$ will be look like the restriction of $\alpha$ to $I_n$ with $c$ skipped. Thus, any colour $d$ which is dense somewhere for $\beta$ will be dense somewhere inside $I_n$ for $\alpha$ if $d < c$, or if $d \geq c$, then $d + 1$ will be dense. Thus, the reduction works.
\end{proof}
\end{lemma}

\begin{corollary}
\label{corr:etainduction}
\begin{align*}(\cetap)_k & \leqW \crt^1_{k} \times \crt^1_{k-1} \times \ldots \times \crt^1_{2} \\
 & \leqW (\crt^1_2)^{k-1} \times (\crt^1_2)^{k-2} \times \ldots (\crt)^1_2 \\
   & \equivW (\crt^1_2)^{k(k-1)/2} \end{align*}
\end{corollary}

These bounds allow us to characterize the strength of $(\cetap)_k$ by virtue of our considerations of the $k$-finitary part of $\crt_m$.

\begin{corollary}
$(\cetap)_k \equivW \crt^1_k$, whence $\cetap \equivW \crt^1_+$.
\begin{proof}
Corollary \ref{corr:etainduction} tells us that $(\cetap)_k \leqW \crt^1_m$ for some $m \geq k$. As $(\cetap)_k$ has codomain $k$, this already suffices to conclude that $(\cetap)_k \leqW \crt^1_m$ by Corollary \ref{corr:fincrtkm}. Next we show $\crt^1_k \leqW (\cetap)_k$. Fix a computable bijection $\nu : \mathbb{N} \to \mathbb{Q}$. Given a colouring $f : \mathbb{N} \to k$ as input for $\crt^1_k$, we define the colour $\alpha_f : \mathbb{Q} \to k$ by $\alpha_f(q) = f(\nu^{-1}(q))$. Clearly, any colour appearing somewhere dense in $\alpha_f$ must have appeared infinitely often in $f$.
\end{proof}
\end{corollary}

Our result that $c\etap \equivW \crt^1_+$ stands in contrast to the reverse mathematics results obtained in \cite{FrittaionP17}. In reverse mathematics, $\rt^1_\Nn$ is equivalent to $\mathrm{B}\Sigma_2^0$ \cite{hirstphd}, yet \cite[Theorem 3.5]{FrittaionP17} shows that $\mathrm{B}\Sigma_2^0$ does not imply $\etap$ over $\mathrm{RCA}_0$.


\section{$\art_\bbQ$ and related problems}\label{sec:artq}

We now analyse the logical strength of the principle $\art_\bbQ$.
As in the case of $\weishuffle$, we start with a proof of $\art_\bbQ$
in $\rca+\sztind$.
This will give us enough insights to assess the strength of the corresponding
Weihrauch problems.

\subsection{Additive Ramsey over $\bbQ$ in reverse mathematics}

As a preliminary step, we figure out the strength of $\ort_\bbQ$,
the ordered Ramsey theorem over $\bbQ$.
It is readily provable from $\rca$ and is thus much weaker than most other principles we analyse. We can be a bit more precise by considering $\rca^*$
which is basically the weakening of $\rca$ where induction is restricted to $\Delta^0_1$ formulas 
(see~\cite[Definition X.4.1]{simpson} for a nice formal definition).
\begin{lemma}
  \label{lem:revnath-ortq}
  $\rca\vdash\weiortom_{\bbQ}$. Moreover, if we substitute ``finite poset $P$'' with ``bounded poset $P$'', 
  $\rca^* \vdash \rca \Rightarrow \weiortom_{\bbQ}$
\end{lemma}
\begin{proof}
  We start by proving $\weiortom_{\bbQ}$ in $\rca$. Let $c:[\bbQ]^2\to P$ be an
  orderd colouring, where $(P,\prec_P)$ is a finite poset. Then, $\rca$ proves
  that there exists a $\prec_P$-minimal $p\in P$ such that, for some $x'<y'$ in
  $\bbQ$, $c(x',y')=p$ (this follows, for instance, from bounded $\Sigma^0_1$
  comprehension: consider the set $\{q\in P:\exists x',y'(c(x',y')=q\}$,
  and take any $\prec_P$-minimal element of this set).
  Then, however we take $x<y$ such that $x'\leq x < y\leq y'$,
  we have that $c(x,y)=p$, by $\prec_P$-minimality of $p$ and orderedness of $c$.
  It follows that $]x',y'[$ is $c$-homogeneous (the same would be true for
  $[x',y']$), as we needed.

  We now move to prove the second claim. We point out that this claim plays
  no part in the rest of the paper, and it is only added for completeness.
  We also point out that the use of ``bounded posets'' as opposed to
  ``finite posets'' is due to the fact that the term ``finite'' is not
  well-defined when working over $\rca^*$ (see for instance
  \cite{weaker-cousins-rt-fiori-carones-kolodziejczyk-kowalik} for more on
  the matter): by this, we do not claim that the result would fail for other
  notions of ``finiteness'', we simply selected the notion that
  seems to be the most standard in the literature.
  
  Suppose that $\varphi(i)$ is a $\Sigma^0_1$ formula such that
  $\forall i,j((\varphi(j)\wedge i<j) \rightarrow \varphi(i))$,
  and that $n\in\bbN$ is such that $\forall i(\varphi(i)\rightarrow i<n)$. We
  claim that there exists $m\in\bbN$ such that
  $\forall i(\varphi(i)\leftrightarrow i<m)$.
  By Lemma $2.5$ of \cite{rca0star-simpson-smith},
  this gives $\rca^*\vdash \weiortom_{\bbQ}\Rightarrow \Sigma^0_1\text{-induction}$.

  Suppose that $\varphi(i)$ is $\exists k\theta(i,k)$, for $\theta$ a $\Delta^0_1$
  formula. We define a colouring $c(x,y)$ of pairs of rationals as follows.
  For $x\neq y$, let $h(x,y)$ be the largest integer $h$ such that $|x-y|<2^{-h}$.
  For rationals $x<y$, we let $c(x,y)$ be the least $i\leq n$ such that
  $\forall k<h(x,y)\neg \theta(k,i)$. Let us consider the poset
  $([0,n],\geq)$: then, $c:[\bbQ]^2\to [0,n]$ is ordered, since $x'\leq x < y
  \leq y'$ implies $h(x,y)\geq h(x',y')$, and so the minimal $i\leq n$ such
  that $\forall k<h(x,y)\neg \theta(k,i)$ cannot be $\leq_{\bbN}$-smaller than
  the minimal $i\leq n$ such that $\forall k<h(x',y')\neg \theta(k,i)$.

  By $\weiortom_{\bbQ}$, there is a $c$-homogeneous interval $]x',y'[$. Let $m$
  be the colour for which $]x',y'[$ is homogeneous. It is easy to see that this
  $m$ is the bound we were looking for. 
\end{proof}

We now show that the shuffle principle is equivalent to $\art_\bbQ$. So overall,
much like the Ramsey-like theorems of~\cite{KMPS19}, they are equivalent to
\sztindtxt.

\begin{lemma}\label{lem:ind-artq}
    $\rca+\weishuffle\vdash \art_\bbQ$. Hence, $\rca+\sztind\vdash\art_\bbQ$.
\end{lemma}

\begin{proof}
Fix a finite semigroup $(S, \cdot)$ and an additive colouring $c : [\bbQ]^2 \to S$. Say a colour $s\in S$ \emph{occurs} in $X \subseteq \bbQ$ if there exists $(x,y) \in [X]^2$ such that $c(x,y) = s$.

We proceed in two stages: first, we find an interval $\ooint{u, v}$ such that all colours occurring in $\ooint{u, v}$ are $\GrJ$-equivalent to one another.
Then we find a subinterval of $\ooint{u, v}$ partitioned into finitely many dense homogeneous sets.
For the first step, we apply the following lemma to obtain a subinterval $I_1 = {]u,v[}$ of $\bbQ$ where
all colours lie in a single $\GrJ$-class.

\begin{lemma}\label{lem:interm-ortq}
For every additive colouring $c$, there exists $(u,v) \in [\bbQ]^2$
such that all colours of $\restr{c}{\ooint{u,v}}$ are $\GrJ$-equivalent to one another.
\end{lemma}
\begin{proof}
If we post-compose $c$ with a map taking a semigroup element to its $\GrJ$-class, we get an ordered colouring. Applying $\weiortom_\bbQ$ yields a suitable
  interval.
\end{proof}

  Moving on to stage two of the proof, we want to look for a subinterval of
  $I_1$ partitioned into finitely many dense homogeneous sets.
  To this end, define a colouring $\gamma: I_1 \to S^2$ by setting $\gamma(z) = (c(u,z),c(z,v))$.

By $\weishuffle$, there exist $x,y \in I_1$ with $x < y$ such that
$]x,y[$ is a $\gamma$-shuffle.
For $l, r \in S$, define $H_{l,r} \eqdef \gamma^{-1}(\{(l,r)\}) \cap \ooint{x, y}$; note that this is a set by bounded recursive comprehension.
Clearly, all $H_{l,r}$ are either empty or dense in $\ooint{x,y}$, and moreover $\ooint{x, y} = \bigcup_{l,r} H_{l, r}$. Since there are finitely many
  pairs $(l,r)$, all we have to prove is that each non-empty $H_{l,r}$ is homogeneous for $\colA$.

Let $s = c(w,z)$ such that $w, z \in H_{l,r}$ with $w < z$. By additivity of $c$ and the definition of $H_{l,r}$,
\begin{equation} \label{eq:ramQR} s \cdot r = c(w,z) \cdot c(z,v) = c(w, v)  = r. \end{equation}
In particular $r \le_{\GrR} s$.
But we also have $r \GrJ s$, which gives $r \GrR s$ by Lemma~\ref{lem:leRJR}.
This shows that all the colours occurring in $H_{l,r}$ are $\GrR$-equivalent to one another.
A dual argument shows that they are all $\GrL$-equivalent, so they are all $\GrH$-equivalent.
The assumptions of Lemma~\ref{lem:Hidemgroup} are satisfied, so their $\GrH$-class is actually a group.

All that remains to be proved is that any colour $s$ occurring in $H_{l,r}$ is actually the (necessarily unique) idempotent of this $\GrH$-class.
Since $r \GrR s$, there exists $a$ such that $s = r \cdot a$. But then by~(\ref{eq:ramQR}), $s\cdot s = s\cdot r\cdot a = r\cdot a = s$, so $s$ is necessarily the idempotent.
Thus, all sets $H_{l,r}$ are homogeneous and we are done.
\end{proof}

We conclude this section by showing that the implication proved in the Lemma above reverses., thus giving the precise strength of $\art_\bbQ$ over $\rca$.

\begin{theorem}\label{thm:shuffle-artq}
    $\rca + \art_\bbQ\vdash\weishuffle$. Hence, $\rca \vdash \art_\bbQ \leftrightarrow \sztind$.
\end{theorem}
\begin{proof}
    Let $f\colon\bbQ\to n$ be a colouring of the rationals. Let $(S_n,\cdot)$ be the finite semigroup defined by $S_n=n$ and
    $a \cdot b = a$ for every $a,b\in S_n$.
    Define the colouring $c\colon [\bbQ]^2\to S_n$ by setting $c(x,y)=f(x)$ for every $x,y\in \bbQ$.
    Since for every $x<y<z$, $c(x,z)=f(x)=c(x,y)\cdot c(y,z)$, $\colA$ is additive.
    By additive Ramsey, there exists $]u,v[$ which is $c$-densely homogeneous
    and thus a $f$-shuffle.
\end{proof}

\subsection{Weihrauch complexity of additive Ramsey}

We now start the analysis of $\art_\bbQ$ in the context of
Weihrauch reducibility. We will mostly summarize results, relying on the intuitions
we built up so far.
First off, we determine the Weihrauch degree of the ordered Ramsey theorem over $\bbQ$.

\begin{theorem}
  \label{thm:wei-ort}
    Let $\ort_\bbQ$ be the problem whose instances are ordered colourings $c:[\bbQ]^2\to P$, for some finite poset $(P,\prec)$, and whose possible outputs on
    input $c$ are intervals on which $c$ is constant. We have that $\ort_\bbQ\equivW\lpo^*$.
\end{theorem}

\begin{proof}
    $\lpo^*\leqsW\ort_\bbQ$: let $\langle n,p_0,\dots,p_{n-1}\rangle$ be an instance of $\lpo^*$. Let $(P,\prec)$ be the poset such that $P=2^n$, the set of subsets of $n$, and $\prec \, = \, \supset$, i.e.\ $\prec$ is reverse inclusion.

    We define an ordered colouring $c:[\bbQ]^2\to P$ in stages by deciding, at stage $s$, the colour of all the pairs of points $(x,y)\in [\bbQ]^2$ such that $|x-y|>2^{-s}$.

    At stage $0$, we set $c(x,y)=\emptyset$ for every $(x,y)\in [\bbQ]^2$ such that $|x-y|>1$. At stage $s>0$, we check $\restr{p_i}{s}$ for every $i<n$ (i.e., for every $i$, we check the sequence $p_i$ up to $p_i(s-1)$), and for every $(x,y)\in [\bbQ]^2$ with $2^{-s+1}\geq |x-y|> 2^{-s}$, we let
    \[
        c(x,y)= \{ i<n: \exists t< s (p_i(t)=1) \}.
    \]
    It is easily seen that $c$ defined as above is an ordered colouring: if $x\leq x'<y'\leq y'$, then $|x'-y'|\leq |x-y|$, which means that to determine the colour of $(x',y')$ we need to examine a longer initial segment of the $p_i$s. Let $I\in \ort_\bbQ(P,c)$, and let $\ell\in\bbN$ be least such that the length of $I$ is larger that $2^{-\ell}$: since $I$ is $c$-homogeneous, we know that for every $i<n$, $\exists t (p_i(t)=1) \Leftrightarrow \exists t<\ell (p_i(t)=1)$.
    Hence, for every pair of points $(x,y)\in [I]^2$, the colour of $c(x,y)$ is exactly the set of $i$ such that $\lpo(p_i)=1$.

    $\ort_\bbQ\leqW\lpo^*$: Let $(P,c)$ be an instance of $\ort_\bbQ$, for some finite poset 
    $(P,\prec_P)$. Let $<_L$ be a linear extension of $\prec_P$, and notice that $c:\bbQ\to (P,<_L)$ is still an ordered colouring. Let $r_0 <_L r_1 <_L \dots <_L r_{|P|-1}$ be the elements of $P$. The idea of the proof is to have one instance of $\lpo$ per element of $P$, and to check in parallel the intervals of the rationals to see if they are $c$-homogeneous for the corresponding element of $P$. Anyway, one has to be careful as to how these intervals are chosen: to give an example, if we find that a certain interval $I$ is not $c$-homogeneous for the $<_L$-maximal element $r_{|P|-1}$, because we found, say, $x<y$ such that $c(x,y)\neq r_{|P|-1}$, then not only do we flag the corresponding instance of $\lpo$ by letting it contain a $1$, but we also restrict the research of all the other components so that they only look at intervals \emph{contained in $]x,y[$}. By proceeding similarly for all the components, since $c$ is ordered, we are sure that we will eventually find a $c$-homogeneous interval.

    We define the $|P|$ instances $p_0,p_1,\dots,p_{|P|-1}$ of $\lpo$ in stages as follows. Let $a_n$ be an enumeration of the ordered pairs of rationals, i.e.\ an enumeration of $[\bbQ]^2$, with infinitely many repetitions. At every stage $s$, some components $i$ will be ``active'', whereas the remaining components will be ``inactive'': if a component $i$ is inactive, it can never again become active. Moreover, at every stage $s$, there is a ``current pair'' $a_{n_s}$ and a ``current interval'' $a_{m_s}$ (for this proof, it is practical to see ordered pairs of rational as both pairs and as denoting extrema of an open interval). We begin stage $0$, by putting the current pair and the current interval equal to $a_0$. Moreover, every component is set to be active.

    At stage $s$, for every inactive component $j<|P|$, we set $p_j(s)=1$. For every active component $i$, there are two cases:
    \begin{itemize}
        \item if, for every active component $i$, $c(a_{n_s})\geq_L r_i$, then we look for the smallest $\ell>n_s$ such that $a_{\ell}\subseteq a_{m_s}$ (i.e., we look for a pair of points contained in the current interval), and set $a_{n_{s+1}}=a_{\ell}$, and $a_{m_{s+1}}=a_{m_s}$. We set $p_i(s)=0$ and no component is set to inactive. We then move to stage $s+1$.
        \item suppose instead there is an active component $i$ such that $c(a_{n_s})<_L r_i$: let $i$ be the minimal such $i$, then we set every $j\geq i$ to inactive (the ones that were already inactive remain so) and we let $p_j(s)=1$. We then let $a_{m_{s+1}}=a_{n_s}$, and we look for the least $\ell>{n_s}$ such that $a_{\ell}\subset a_{n_s}$: we set $a_{n_{s+1}}=a_{\ell}$, and we set $p_k(s)=0$ for every active component $k<|P|$. We then move to stage $s+1$.
    \end{itemize}
    We iterate the procedure above for every integer $s$.

    Let $\sigma \in 2^{|P|}$ be such that $\sigma\in \lpo^*(\langle |P|, p_0,\dots,p_{|P|-1}\rangle)$. Notice that $\sigma(0)=0$, since no pair of points can attain colour $<_L$-below $r_{0}$. Moreover, notice that $\sigma(i)=0$ if and only if the component $i$ was never set inactive. Hence, let $i$ be maximal such that $\sigma(i)=0$, and let $t$ be a state such all components $j>i$ have been set inactive by step $t$. Hence, after step $t$, the current interval $I$ never changes, and thus we eventually check the colour of all the pairs in that interval. Since the second case of the construction is never triggered, it follows that $I$ is a $c$-homogeneous interval. Hence, in order to find it, we know we just have to repeat the construction above until all the components of index larger than $i$ are set inactive. This proves that $\ort_\bbQ\leqW \lpo^*$.
\end{proof}

Now let us discuss Weihrauch problems corresponding to $\art_\bbQ$.

\begin{definition}
\label{def:sartbbq}
Regard $\sart_\bbQ$ as the following Weihrauch problem:
  the instances are pairs $(S,c)$ where $S$ is a finite semigroup and
  $c:[\bbQ]^2\to S$ is an additive colouring of $[\bbQ]^2$, and such that, for
  every $C\subseteq S$ and every interval $I$ of $\bbQ$, $(I,C)\in \sart_\bbQ$
  if and only if $I$ is $c$-densely homogeneous for the colours of $C$.
\\\\
Similarly to what we did in \Cref{def:wei-weak-shuffle}, we also
introduce the problems $\cart_\bbQ$ and $\iart_\bbQ$
that only return the set of colours and the interval respectively.
\end{definition}


We start by noticing that the proof of \Cref{thm:shuffle-artq} can be readily adapted to show the following.
\begin{lemma}\label{lem:artq-lower-bound}
    \begin{itemize}
        \item $\cshuffle\leqsW \cart_\bbQ$, hence $(\lpo')^*\leqW \cart_\bbQ$.
        \item $\ishuffle\leqsW\iart_\bbQ$, hence $\tcn^*\leqW \iart_\bbQ$.
        \item $\sshuffle\leqsW\sart_\bbQ$, hence $(\lpo')^*\times\tcn^*\leqW \sart_\bbQ$.
    \end{itemize}
\end{lemma}

The rest of the section is devoted to find upper bounds for $\cart_\bbQ$, $\iart_\bbQ$ and $\sart_\bbQ$.
The first step to take is a careful analysis of the proof of \Cref{lem:ind-artq}. For an additive colouring $c\colon[\bbQ]^2\to S$, the proof can be summarized as follows:
\begin{itemize}
    \item we start with an application of $\ort_\bbQ$ to find an interval $]u,v[$ such that all the colours of $\restr{c}{]u,v[}$ are all $\GrJ$-equivalent (\Cref{lem:interm-ortq}).
    \item define the colouring $\gamma\colon \bbQ\to S^2$ and apply $\weishuffle$ to it, thus obtaining the interval $]x,y[$.
    \item the rest of the proof consists simply in showing that $]x,y[$ is a $c$-densely homogeneous interval.
\end{itemize}
Hence, from the uniform point of view, this shows that $\art_\bbQ$ can be computed via a composition of $\weishuffle$ and $\ort_\bbQ$. Whence the next theorem.

\begin{theorem}
  \label{thm:fin-theo}
    \begin{itemize}
    \item $\cart_\bbQ\leqW (\lpo')^*\times \lpo^*$, therefore $\cart_\bbQ\equivW (\lpo')^*$.
        \item $\iart_\bbQ\leqW \tcn^*\times\lpo^*$, therefore $\iart_\bbQ\equivW \tcn^*$.
        \item $\sart_\bbQ\leqW (\lpo')^*\times\tcn^*\times\lpo^*$, therefore $\sart_\bbQ\equivW (\lpo')^*\times\tcn^*$.
    \end{itemize}
\end{theorem}

\begin{proof}
  The three results are all proved in a similar manner. We recall that $\lpo^*\leqW\cn$ and observe that $\lpo^*$ is single-valued.
  This enables us to use \Cref{lem:comp-to-paral-prod} with $\lpo^*$ in place of $\weipr P$.

  For $\mathsf{x}\in \{\mathsf{c},\mathsf{i},\mathsf{s}\}$ and every $n\in\bbN$, let
    $\mathsf{x}\art_{\bbQ,n}$ be the restriction of $\mathsf{x}\art_{\bbQ}$ to instances of the form $(S,c)$ with $S$ of cardinality $n$. Hence, by the
    considerations preceding the statement of the theorem in the body of the paper, we have the following facts:
    \begin{itemize}
        \item $\cart_{\bbQ,n}\leqW \cshuffle_{n^2}* \ort_\bbQ$, hence, by \Cref{lem:cshuffle-leq-lpo} and \Cref{thm:wei-ort}, we have that $\cart_{\bbQ,n}\leqW (\lpo')^{2^{n^2}-1}*\lpo^*$. By \Cref{lem:comp-to-paral-prod}, we have that $\cart_{\bbQ,n}\leqW (\lpo')^{2^{n^2}-1}\times\lpo^*$, from which the claim follows.
        \item $\iart_{\bbQ,n}\leqW \ishuffle_{n^2}*\ort_\bbQ$, hence, by \Cref{lem:ishuffle-leq-tcn} and \Cref{thm:wei-ort}, we have that $\iart_{\bbQ,n}\leqW \tcn^{n^2-1}*\lpo^*$. By \Cref{lem:comp-to-paral-prod}, we have that $\iart_{\bbQ,n}\leqW \tcn^{n^2-1}\times\lpo^*$, from which the claim follows.
        \item $\sart_{\bbQ,n}\leqW \sshuffle_{n^2}*\ort_\bbQ$, hence, by \Cref{lem:sshuffle-leq-tcnlpo} and \Cref{lem:comp-to-paral-prod}, we have that $\sart_{\bbQ,n}\leqW (\lpo'\times\tcn)^{2^{n^2}-1}*\lpo^*$. By \Cref{lem:comp-to-paral-prod}, we have that  $\sart_{\bbQ,n}\leqW (\lpo'\times\tcn)^{2^{n^2}-1}\times\lpo^*$, from which the claim follows.
        \end{itemize}

\end{proof}


\section{$\art_\bbN$ and $\ort_\bbN$}\label{sec:artN}
\label{sec:artn}

We finally turn to the case of the additive and ordered theorems over $\bbN$ and prove \Cref{thm:wei-eqs-bbN}.
We obtain results which are completely analogous to
the case of $\bbQ$ when it comes to the additive Ramsey theorem. However, in contrast to Theorem \ref{thm:wei-ort}, the ordered Ramsey theorem for $\bbN$ exhibits the same behaviour as the additive Ramsey theorem.

That the principles $\ort_\bbN$ and $\art_\bbN$ are equivalent to \sztindtxt~
was established in~\cite{KMPS19}, so we only focus on the analysis of the
Weihrauch degrees below. We first start by defining properly the principles
involved, and then we give the proof that $\tcn^*$, $(\lpo')^*$ or their product
reduces to them. We then give the converse reductions, first for the
principles pertaining to the ordered colourings, and then we handle the additive
colourings.
The proof for the ordered colouring is a simple elaboration on~\cite[Lemma 4.3]{KMPS19}.
For the additive colouring, formally the corresponding statement in that paper,~\cite[Proposition 4.1]{KMPS19},
depends on the ordered version in a way that would translate to a composition
in the setting of Weihrauch degrees. It turns out that we can avoid
invoking the composition by carefully interleaving the two steps in our analysis.

\subsection{Definitions}

We have already covered the principles $\ort_\bbN$ and $\art_\bbN$ in~\Cref{sec:background}.
The corresponding Weihrauch problems are relatively clear: given a colouring as
input, as well as the finite semigroup or finite ordered structure, output an infinite homogeneous
set. The principles $\cort_\bbN$ and $\cart_\bbN$ instead only output a possible
colour for an infinite homogeneous set -- much like in the case for $\bbQ$.
However, the principles $\iort_\bbN$ and $\iart_\bbN$ will require some more
attention; now it is rather meaningless to ask for a containing interval.
Nevertheless, the analogous principle will also output some information
regarding the possible location of an homogeneous set, without giving away a
whole set or a candidate colour, so we keep a similar naming convention.

\begin{definition}
  \label{def:ramN-wei}
Define the following Weihrauch problems:
  \begin{itemize}
    \item $\ort_\bbN$ takes as input a finite poset $(P, \preceq_P)$ and
      a right-ordered colouring $c : [\bbN]^2 \to P$, and outputs an infinite
      $c$-homogeneous set $\subseteq \bbN$.
    \item $\art_\bbN$ takes as input a finite semigroup $S$ and
      an additive colouring $c : [\bbN]^2 \to S$, and outputs an infinite
      $c$-homogeneous set $\subseteq \bbN$.
    \item $\cort_\bbN$ takes as input a finite poset $(P, \preceq_P)$ and
      a right-ordered colouring $c : [\bbN]^2 \to P$, and outputs a colour $p \in P$
      such that there exists an infinite $c$-homogeneous set $\subseteq \bbN$
      with colour $p$.
    \item $\cart_\bbN$ takes as input a finite semigroup $S$ and
      an additive colouring $c : [\bbN]^2 \to S$, and outputs a colour $s \in S$
      such that there exists an infinite $c$-homogeneous set $\subseteq \bbN$
      with colour $s$.
    \item $\iort_\bbN$ takes as input a finite poset $(P, \preceq_P)$ and
      a right-ordered colouring $c : [\bbN]^2 \to P$, and outputs a $n_0 \in \bbN$
      such that there is an infinite $c$-homogeneous set $X \subseteq \bbN$
      with two elements $\le n_0$.
    \item $\iart_\bbN$ takes as input a finite semigroup $S$ and
      an additive colouring $c : [\bbN]^2 \to S$, and outputs a $n_0 \in \bbN$
      such that there is an infinite $c$-homogeneous set $X \subseteq \bbN$
      with two elements $\le n_0$.
  \end{itemize}
\end{definition}

\subsection{Reversals}
\label{subsec:aortN-reversal}

\begin{lemma}
  \label{lem:iaortN-reversal}
We have $\ect \lesW \iort_\bbN$ and $\ect \lesW \iart_\bbN$.
\end{lemma}
\begin{proof}
Let $f : \bbN \to k$ be a would-be instance of $\ect$.
  Then one may define the colouring $\tilde{f} : [\bbN]^2 \to \pow(k)$
  by setting $a \in \tilde{f}(n,m)$ if and only if there is $n'$ with
  $n \le n' \le m$ and $f(n') = a$.
  This colouring is both additive for the semigroup $(\pow(k), \cup)$ and
  ordered by $\subseteq$, and can be fed to either $\iort_\bbN$ or $\iart_\bbN$.
Let $n_0$ be such that there is an infinite $\tilde{f}$-homogeneous set with
first two elements $k_0 < k_1 \le n_0$. Clearly, every colour occuring in $f$
  after $n_0$ needs to occur in $\tilde{f}(k_0,k_1)$; so $n_0$ is a solution of
  the given instance for $\ect$.
\end{proof}

\begin{lemma}
$\cort_\bbN^* \equivW \cort_\bbN$ and $\cart_\bbN^* \equivW \cart_\bbN$
\begin{proof}
The non-trivial
  reductions are easily made by amalgamating finite sequences of colouring via
  a pointwise product, which will always still carry an additive or ordered
  structure.
\end{proof}
\end{lemma}

\begin{lemma}
  \label{lem:caortN-reversal}
$\lpo' \lesW \cort_\bbN$ and $\lpo' \lesW \cart_\bbN$.
\end{lemma}
\begin{proof}
  We use $\sisfin$ in place of $\lpo'$. We start with an input $f : \bbN \to 2$ for $\sisfin$. We compute
  $\tilde{f} : [\bbN]^2 \to 2$ where $\tilde{f}(n,m) = 1$ iff $1 \in f^{-1}([n,m])$. This yields an additive and ordered colouring.
  The colour of any given $\tilde{f}$-homogeneous set indicates if $f$
  has infinitely many ones or not, thus answering $\sisfin$ for $f$.
\end{proof}

\subsection{
Reducing the ordered Ramsey theorem over $\bbN$ to $(\lpo')^*$ and $\ect$}

We now explain how to bound the Weihrauch degree of $\ort_\bbN$ and its weakenings.
To do so, it will be helpful to consider a construction approximating would-be
homogeneous sets for a given right-ordered colouring $c : [\bbN]^2 \to P$
and a target colour $p \in P$. With these parameters, we build a recursive
sequence of finite sets $Y^{(p)} : \bbN \to \powfin(\bbN)$ meant to approximate a
$p$-homogeneous set (we shall simply write $Y$ instead of $Y^{(p)}$ when $p$
may be inferred from context).
If the construction succeeds, $\limsup(Y)$ will be an
infinite homogeneous set with colour $p$, otherwise $\limsup(Y)$ will be finite.
But the important aspect will be that a fixed number of calls to $(\lpo')^*$
will let us know if the construction was successful or not, while $\ect$ can
indicate after which indices $n$ we shall have $Y_n \subseteq \limsup(Y)$ when
it succeeds.

Now let us describe this construction for a fixed $c$ and $p$.
We begin with $Y_0 = \emptyset$ and will maintain the invariant that
$\max(Y_n) < n$ and for every $(k,k') \in [Y_n]^2$, $c(k,k') = p$.
Then, for $Y_{n+1}$, we have several possibilities;
\begin{itemize}
\item If $\min(Y_n)$ exists and
  for any $\min(Y_n) \le k < n$, we have that $p \prec_P c(k,n)$, we set $Y_{n+1} = \emptyset$
    and say that the construction was \emph{injured} at stage $n$.
  \item Otherwise, if we have some $k' < n$ such that $c(k',n) = p$
    and, for every $k \in Y_n$, $k < k'$ and $c(k,k') = p$, then we set
    $Y_{n+1} = Y_n \cup \{k\}$ and say that the construction \emph{progressed}
    at stage $n$.
  \item Otherwise, set $Y_{n+1} = Y_n$ and say that the construction \emph{stagnated}.
\end{itemize}

Clearly, we can also define recursive sequences
$\mathrm{injury}^{(p)} : \bbN \to 2$
and $\mathrm{progress}^{(p)} : \bbN \to 2$
that witness whether the construction
was injured or progressed, and we have that $\limsup(Y)$ is infinite if and only
if $\mathrm{injury}$ contains finitely many $1$s and $\mathrm{progress}$ contains
infinitely many $1$s. $\limsup(Y)$ is moreover always $c$-homogeneous with colour
$p$.

\begin{lemma}
  \label{lem:ort-build}
  For any ordered colouring $c$, there is $p$ such that $\limsup(Y)$ is infinite
\end{lemma}
\begin{proof}
  The suitable $p$ may be found as follows: say that a colour $p$ occurs after $n$
  in $c$ if there is $k > m \ge n$ with $c(m,k) = p$.
  There is a $n_0$ such that every colour occuring after $n_0$ in $c$ occurs arbitrarily
  far. For the $\preceq_P$-maximal such colour occuring after $n_0$, the construction
  above will succeed with no injuries after stage $n_0$ and infinitely many progressing
  steps (this is exactly the same argument as for~\cite[Lemma 4.3]{KMPS19}).
\end{proof}

\begin{lemma}
\label{lem:cortN-below}
We have that $\cort_\bbN \lesW (\lpo')^*$.
\end{lemma}
\begin{proof}
Given an input colouring $c$, compute in parallel all $\mathrm{injury}^{(p)}$ and
  $\mathrm{progress}^{(p)}$ for every colour $p$ and feed each sequence to an
  instance of $\lpo'$.
  By~\Cref{lem:ort-build}, there is going to be some $p$
  for which there are going to be finitely many injuries and infinitely many progressing
  steps, and that $p$ is the colour of some homogeneous set.
\end{proof}

\begin{lemma}
\label{lem:ortN-below}
We have that $\ort_\bbN \leW (\lpo')^* \times \ect$.
\end{lemma}
\begin{proof}
Given an input colouring $c$, compute in parallel all $\mathrm{injury}^{(p)}$ and
  $\mathrm{progress}^{(p)}$ for every colour $p$ and feed each sequence to an
  instance of $\lpo'$ \emph{and} all $\mathrm{injury}^{(p)}$ to $\ect$.
As before, use $\lpo'$ to find out some $p$ for which the construction succeed.
  For that $p$, $\ect$ will yield some $n_0$ such that
  $\mathrm{injury}^{(p)}_n = 0$ for every $n \ge n_0$, so in particular,
  $\limsup(Y^{(p)}) = \bigcup_{n \ge n_0} Y_n^{(p)}$, which is computable from
  $n_0$.
\end{proof}

\begin{lemma}
\label{lem:iortN-below}
We have that $\iort_\bbN \lesW \ect$.
\end{lemma}
\begin{proof}
Given an input colouring $c$, consider for every colour $p$ the sequence
  $u^{(p)} : \bbN \to \{0,1,2\}$ defined by $u^{(p)}_n = \min(3,|Y_n|)$.
  Clearly it is computable from $c$. Applying $\ect$ we get some $n_p$ such that
  \begin{itemize}
    \item either there are infinitely many injuries after $n_p$
    \item or $u^{(p)}_{k} = u^{(p)}_{n_p}$ for every $k \ge n_p$
  \end{itemize}
  By Lemma \ref{lemma:ectandtcn*}, a single instance of $\ect$ can provide the finitely many answers for the finitely many colours.
  
  By~\Cref{lem:ort-build}, we even know there is a $p_0$
  such that $\limsup(Y^{(p_0)})$ is infinite; additionally we defined $u$ in such
  a way that necessarily, $n_{p_0}$ bounds two elements of $\limsup(Y^{(p_0)})$
  because we shall have $u^{(p_0)}_k = 2$ for every $k \ge n_{p_0}$.
So we may simply take the maximum of all $n_p$ to solve our instance of
$\iort_\bbN$.
\end{proof}

This concludes our analysis of the ordered Ramsey theorem.

\subsection{
Reducing the additive Ramsey theorem over $\bbN$ to $(\lpo')^*$ and $\ect$}

We now turn to $\art_\bbN$.
The basic idea is that, given an additive colouring $c$, it is
useful to define the composite colouring $L \circ c$, with $L$ being a map from
a finite semigroup to its $\GrL$-classes. $\le_{\GrR}$ then induces a right-ordered
structure on the colouring.
Constructing a $L \circ c$ homogeneous set $X$ such that we additionally have
that $c(\min X, x) = c(\min X, y)$ for every $x,y \in X \setminus \{\min X\}$
ensures that $X$ is $c$-homogeneous by~\Cref{lem:leLRH}.
So we will give a recipe to construct exactly such an approximation, similarly
to what we have done in the previous section.

So this time around, assume a semigroup $S$ and a colouring $c : [\bbN]^2 \to S$
to be fixed. For every $s \in S$, we shall define a recursive sequence of sets
$Y^{(s)} : \bbN \to \powfin(\bbN)$ (we omit the superscript when clear from context)
such that $\max(Y_n) < n$ and $Y_n \setminus \{\min(Y_n)\}$
be homogeneous, with, if $Y_n \neq \emptyset$, $c(\min(Y_n), k) = s$ for $k \in Y_n \setminus \{\min(Y_n)\}$.

For $n = 0$, we define $Y_0 = \emptyset$. For $Y_{n+1}$, we have a couple of
options:
\begin{itemize}
  \item If $Y_n$ is empty and there is $k < n$ such that $c(k,n) = s$ and there is
    no $k \le k' < n' \le n$ with $c(k',n') <_{\GrR} s$, then set $Y_{n+1} = \{k\}$
    for the minimal such $k$ and say that the construction \emph{(re)starts}.
  \item If $Y_n$ is non-empty and there is some $k$ with $\min(Y_n) \le k < n$
    with $c(k,n) <_{\GrR} s$, set $Y_{n+1} = \emptyset$ and say that the construction
    was \emph{injured} at stage $n$.
  \item Otherwise, if $Y_n$ is non-empty and we have some $\min(Y_n) < k < n$ with
    $c(\min(Y_n),k) = s$ and $c(k,n) \GrR s$, set $Y_{n+1} = Y_n \cup \{k\}$ and
    say that the construction \emph{progresses}.
  \item Otherwise, set $Y_{n+1} = Y_n$ and say that the construction stagnates.
\end{itemize}

We can define auxiliary binary sequences $\mathrm{injury}^{(s)}$ and
$\mathrm{progress}^{(s)}$ that witness the relevant events, and an infinite homogeneous
subset will be built as long as we have finitely many injuries and infinite progress.

\begin{lemma}
  \label{lem:art-sound}
  If $\mathrm{injury}^{(s)}$ has finitely many 1s and $\mathrm{progress}^{(s)}$
  has infinitely many 1s, then $X = \limsup(Y) \setminus \{\min(\limsup(Y))\}$ is a $c$-homogeneous
  infinite set. Furthermore the colour of $X$ is computable from $s$.
\end{lemma}
\begin{proof}
That the condition is sufficient for $X$ to be infinite is obvious; it only
remains to show it is homogeneous.
Note that all elements in $\limsup(Y)$ are necessarily $\GrR$-equivalent to
  one another. Call $y_0 = \min(\limsup(Y))$.
  For $(l,m) \in [X]^2$, we necessarily have $s \le_{\GrL} c(l,m)$. So by~\Cref{lem:leLRH},
  we necessarily have $c(l,m) \GrH s$. Additionally, we know they belong to a
  $\GrH$-class which is a group, so by algrebraic manipulations, we have that
  $[X]^2$ is monochromatic and the corresponding colour is the neutral element of that group.
\end{proof}

\begin{lemma}
  \label{lem:art-build}
  There is some $s$ such that $\mathrm{injury}^{(s)}$ has finitely many 1s and $\mathrm{progress}^{(s)}$
  has infinitely many 1s.
\end{lemma}
\begin{proof}
  Consider, much as we did in the proof of~\Cref{lem:ort-build}, a $n_0$ such
  that all $\GrR$-classes occurring after $n_0$ occur arbitrarily far. Consider
  a minimal such $\GrR$-class $R$. 
  For the minimal $k_0$ such that no $\GrR$-class strictly lower than $R$
  occurs after $k_0$, there is some $s$ such that the set $\{ n \mid c(k_0,n) = s\}$
  is infinite; but, by construction, this set is exactly
  $\limsup(Y^{(s)}) \setminus \{k_0\}$.
\end{proof}

With these two lemmas in hand, it is easy then to carry out a similar analysis
as in the last subsection. We do not expand the proof, which are extremely
similar, but only summarize the results.

\begin{lemma}
  \label{lem:artN-below}
  We have the following reductions:
  \begin{itemize}
    \item $\art_\bbN \leW (\lpo')^* \times \ect$
    \item $\cart_\bbN \lesW (\lpo')^*$
    \item $\iart_\bbN \lesW \ect$
  \end{itemize}
\end{lemma}

\subsection{Wrapping things up: deriving Theorem~\ref{thm:wei-eqs-bbN}}

Let us now combine the lemmas above to prove each item
of Theorem~\ref{thm:wei-eqs-bbN}:
\begin{itemize}
  \item $\ort_\bbN \equivW \art_\bbN \equivW \tcn^* \times (\lpo')^*$: by combining the trivial observation that
  we have $\cart_\bbN \times \iart_\bbN \leW \art_\bbN \times \art_\bbN$ and all of the
    lemmas of Subsection~\ref{subsec:aortN-reversal}, we have
    that $\tcn^* \times \ect \leW \art_\bbN$, which is hald the equivalence
    required (recall that $\tcn^* \equivW \ect$ (Lemma~\ref{lemma:ectandtcn*})).
    The reversal that improves this to an equality is given by the first
    item in Lemma~\ref{lem:artN-below}.
    The equivalence between $\tcn^* \times (\lpo')^*$ and $\ort_\bbN$ is
    established similarly, the only difference being that the reversal is
    obtained by Lemma~\ref{lem:ortN-below}.
  \item $\cort_\bbN \equivW \cart_\bbN \equivW (\lpo')^*$:
    $\cart_\bbN \leW (\lpo')^*$ is a conclusion of Lemma~\ref{lem:artN-below},
    $\cort_\bbN \leW (\lpo')^*$ is from Lemma~\ref{lem:cortN-below}
    while the reversals are given by Lemma~\ref{lem:caortN-reversal}.
  \item $\iort_\bbN \equivW \iart_\bbN \equivW \tcn^*$: similarly,
    Lemmas~\ref{lem:artN-below} and~\ref{lem:iaortN-reversal} establish
    that $\iart_\bbN \equivW \tcn^*$, and the equivalence
    $\iort_\bbN \equivW \tcn^*$ is proven using Lemmas~\ref{lem:iortN-below} and
    ~\ref{lem:iaortN-reversal}.
\end{itemize}

\section{How the colours are coded}
\label{sec:colours}
All principles we have studied that receive as input a colouring of some sort also receive explicit finite information about the finite set/finite poset/finite semigroup of colours. This is not the approach we could have taken: in the case of a plain set of colours, the colouring itself contains the information on how many colours it is using. In the cases where the colours carry additional structure, this could have been provided via the atomic diagram of the structure. This would lead to the requirement that only finitely many colours are used to be a mere promise.

We will first demonstrate the connection between the two versions on a simple example, namely $\crt^1_+$. Let us denote with $\crt^1_\mathbb{N}$ the principle that takes as input a colouring $\alpha : \mathbb{N} \to \mathbb{N}$ such that the range of $\alpha$ is finite, and outputs some $n \in \mathbb{N}$ such that $\alpha^{-1}(n)$ is infinite.

\begin{proposition}
 $\crt^1_+ \star \cn \equivW \crt^1_\mathbb{N}$
 \begin{proof}
Instead of $\crt^1_+ \star \cn \leqW \crt^1_\mathbb{N}$ we show that $\crt^1_+ \star \operatorname{Bound} \leqW \crt^1_\mathbb{N}$, where $\operatorname{Bound}$ receives as input an enumeration of a finite initial segment of $\mathbb{N}$, and outputs an upper bound for it. Here is how we produce the input to $\crt^1_\mathbb{N}$ given an input $A$ to $\operatorname{Bound}$ and a sequence $(k_i,\alpha_i)_{i \in \mathbb{N}}$ of partial inputs to $\crt^1_+$: We search for some $k_{i_0}$ to be defined, and then start copying $\alpha_{i_0}$ until $i_0$ gets enumerated into $A$ (if this never happens, $\alpha_{i_0}$ is total and becomes the input to $\crt^1_\mathbb{N}$). Then we search for some $i_1 > i_0$ such that $k_{i_1}$ is defined, and then continue to produce the colouring $\alpha_{i_1} + k_0$; either forever or until $i_1$ gets enumerated into $A$. We repeat this process until some $i_\ell$ is reached which is never enumerated into $A$ (this has to happen).

Given a colour $c$ that appears infinitely often in the resulting colouring, we can retrace our steps and identify what $i_\ell$ was. We can then un-shift $c$ to obtain a colour appearing infinitely often in $\alpha_{i_\ell}$, and thereby answer $\crt^1_+ \star \cn$.

For the converse direction, we observe that $\cn$ can compute from a colouring $\alpha : \mathbb{N} \to \mathbb{N}$ with finite range some $k \in \mathbb{N}$ such that $\alpha$ is a $k$-colouring.
 \end{proof}
 \end{proposition}

The very same relationship holds for all our principles, i.e.~the Weihrauch degree of the version without finite information on the colours is just the composition of the usual version with $\cn$. The core idea, as in the proposition above, is that we can always start over by moving to a fresh finite set of colours. For the interval versions we may have to do a little bit more work to encode the $\cn$-output by ensuring that all ``large'' intervals can never be a valid answer. As a second example, let us consider the variant $\ishuffle_\mathbb{N}$ of $\ishuffle$ where the number of colours is not part of the input.

\begin{proposition}
$\ishuffle \star \cn \equivW \ishuffle_\mathbb{N}$
\begin{proof}
Again, we show $\ishuffle \star \operatorname{Bound} \leqW \ishuffle_\mathbb{N}$. We have an enumeration of a finite initial segment $A$ of $\mathbb{N}$ as input to $\operatorname{Bound}$, together with a sequence $(k_i,\alpha_i)_{i \in \mathbb{N}}$ of partial inputs to $\ishuffle$. We produce a colouring $\alpha : \mathbb{Q} \to \mathbb{N}$ which will use only finitely many colours. We first search for the first $i_0$ for which we learn that $k_{i_0}$ is defined. We start copying $\alpha_i$ to $\alpha$. If $i_0$ is enumerated into $A$ at stage $t$, we assign every rational $\frac{k}{2^{-t}}$ the colour $0$ in $\alpha$. We search for some $i_1 > i_0$ is defined, and then start assigning $\alpha(q) = \alpha_{i_1}(q) + k_{i_0}$. This continues until potentially $i_1$ gets enumerated into $A$, and so forth.

The process has to eventually stop, as $A$ is finite. Moreover, the final $i_{\ell}$ is such that $\alpha_{i_\ell}$ is total. We thus construct a valid input for $\ishuffle_\mathbb{N}$. By considering the length of a resulting interval $I$, we obtain an upper bound for the time $t_{\textrm{final}}$ when we last saw a new $i_j$ number being enumerated into $A$ -- as any such interval cannot contain a rational of the form $\frac{k}{2^{-t_{\textrm{final}}}}$. Knowing an upper bound for $t_{\textrm{final}}$ enables us to find an upper bound $k_{i_\ell}$ for $A$. We also know that restricted to the interval $I$, $\alpha_{i_\ell}$ and $\alpha$ differ by a constant. Thus, $I$ is also a valid output to $\ishuffle(k_{i_\ell}, \alpha_{i_\ell})$.

For the other direction, we again use the fact that $\cn$ can compute from a colouring $\alpha : \mathbb{N} \to \mathbb{Q}$ with finite range some $k \in \mathbb{N}$ such that $\alpha$ is a $k$-colouring.
\end{proof}
\end{proposition}

To see that this observation already fully characterizes the Weihrauch reductions and non-reductions between the usual and the relaxed principles, the notion of a (closed) fractal from \cite{paulybrattka,paulyleroux} is useful.

\begin{definition}
A Weihrauch degree $f$ is called a \emph{fractal}, if there is some $F : \subseteq \baire \rra \baire$ with $f \equivW F$ such that for any $w \in \mathbb{N}^*$ either $w\baire \cap \dom(F) = \emptyset$ or $F|_{w\baire} \equivW f$. If we can chose $F$ to be total, we call the Weihrauch degree a \emph{closed fractal}.
\end{definition}

If $f$ is a fractal and $f \leqW \bigsqcup_{i \in \mathbb{N}} g_i$, then there has to be some $n \in \mathbb{N}$ with $f \leqW g_n$. If $f$ is a closed fractal and $f \leqW g \star \cn$, then already $f \leqW g$. Of our principles, the versions with a fixed number of colours are closed fractals, the versions with a given-but-not-fixed number of colours are not fractals at all, and the versions without explicit colour information are fractals, but not closed fractals. From this, it follows that the versions with no explicit colour information are never Weihrauch equivalent to our studied principles, and that versions without explicit colour information are equivalent to one-another if and only if their counterparts with explicit colour information are equivalent.


\section{Conclusion and future work}

\subparagraph*{Summary} We have analysed the strength of an additive Ramseyan theorem over the
rationals from the point of view of reverse mathematics and found it to be equivalent
to \sztindtxt, and then refined that analysis to a Weihrauch equivalence with
$\tcn^* \times (\lpo')^*$. We have also shown that the problem decomposes nicely:
we get the distinct complexities $(\lpo')^*$ or $\tcn^*$ if we only require either
the set of colours or the location of the homogeneous set respectively.
The same holds true for another equally and arguably more fundamental shuffle principle,
as well as the additive Ramsey theorem over $\bbN$ that was already studied from
the point of view of reverse mathematics in~\cite{KMPS19}.

\subparagraph*{Perspectives} It would be interesting to study further mathematical
theorems that are known to be equivalent to $\sztind$ in reverse mathematics:
this can be considered to contribute to the larger endeavour of studying principles already analyzed
in reverse mathematics in the framework of the Weihrauch degrees. In the particular case
of $\sztind$, it can be interesting to see which degrees are necessary for such an
analysis. We refer to \cite{ramsey-weihrauch-brattka-rakotoniaina} for more on this topic,
and for a more comprehensive study of Ramsey's theorem in the Weihrauch degrees.

Two contemporary investigations in the Weihrauch degrees concerning similar principles to $\etap$ are found in \cite{gill-phd} and \cite{damirreedmanlio}. The former is concerned with the problem of finding a monochromatic copy of $\mathbb{Q}$ as a linear order, given a colouring of $\mathbb{Q}$. The latter studies the problem of finding a monochromatic copy of $2^{<\omega}$ (seen as the structure with the ``prefix of'' predicate) given a colouring of $2^{<\omega}$. The corresponding Weihrauch degrees do differ from the ones explored here, though.

\subsubsection*{Acknowledgements}
The second author warmly thanks Leszek Ko{\l}odziejczyk for the
proof of \Cref{lem:sztind-shuffle} as well as Henryk Michalewski
and Micha{\l} Skrzypczak for numerous discussions on a related project.

\swanseastatement

\bibliography{biblio}

\end{document}